\documentclass[11pt]{article}
\usepackage[utf8]{inputenc}

\usepackage[bookmarks,colorlinks,breaklinks]{hyperref}
\hypersetup{urlcolor=blue, colorlinks=true, citecolor=green!50!black, linkcolor=blue}
\usepackage[letterpaper, left=1in, right=1in, top=0.9in, bottom=0.9in]{geometry}
\usepackage[T1]{fontenc}
\usepackage{bold-extra}
\usepackage[american]{babel}
\usepackage{etoolbox}

\usepackage[normalem]{ulem}

\usepackage{graphicx}
\graphicspath{{images/}{../images/}}

\usepackage{amsmath, amssymb, cases}
\usepackage{theoremstyles}
\usepackage{enumitem}
\usepackage{mdframed}
\usepackage{bbm}
\usepackage{bm}

\usepackage[ruled]{algorithm}
\usepackage[noend]{algpseudocode}

\usepackage{microtype}

\usepackage{mdframed}
\usepackage{xcolor}
\definecolor{palecerulean}{rgb}{0.61, 0.77, 0.89}
\definecolor{pistachio}{rgb}{0.58, 0.77, 0.45}
\definecolor{bubblegum}{rgb}{0.99, 0.76, 0.8}

\definecolor{buublecerulean}{rgb}{0.80,0.77,0.82}
\definecolor{buublepistachio}{rgb}{0.79,0.77,0.63}
\definecolor{ceruleanpistachio}{rgb}{0.59,0.77,0.67}

\usepackage{tikz}
\usetikzlibrary{fit, arrows, shapes, matrix, chains, positioning, backgrounds, arrows.meta, tikzmark, decorations.pathreplacing}

\newcommand{\oraclemin}{\ensuremath{\mathcal{O}_\mathrm{min}}}
\newcommand{\oraclemax}{\ensuremath{\mathcal{O}_\mathrm{max}}}
\newcommand{\oracleds}{\ensuremath{\mathcal{O}_\mathrm{ds}}}
\newcommand{\oraclegamma}{\ensuremath{\mathcal{O}_\Gamma}}

\usepackage{mathrsfs}

\title{Finding Diverse Solutions in Combinatorial Problems with a Distributive Lattice Structure}
\author{Mark de Berg \thanks{Eindhoven University of Technology, Netherlands, \texttt{m.t.d.berg@tue.nl}} \and Andr\'es L\'opez Mart\'inez \thanks{Eindhoven University of Technology, Netherlands, \texttt{a.lopez.martinez@tue.nl}} \and Frits Spieksma \thanks{Eindhoven University of Technology, Netherlands, \texttt{f.c.r.spieksma@tue.nl}}}
\date{\today}

\begin{document}

\maketitle

\pagenumbering{arabic}

\begin{abstract}
    We generalize the polynomial-time solvability of $k$-\textsc{Diverse Minimum s-t Cuts} (De Berg et al., ISAAC'23) to a wider class of combinatorial problems whose solution sets have a distributive lattice structure. We identify three structural conditions that, when met by a problem, ensure that a $k$-sized multiset of maximally-diverse solutions---measured by the sum of pairwise Hamming distances---can be found in polynomial time. We apply this framework to obtain polynomial time algorithms for finding diverse minimum $s$-$t$ cuts and diverse stable matchings. 
    Moreover, we show that the framework extends to two other natural measures of diversity. Lastly, we present a simpler algorithmic framework for finding a largest set of pairwise disjoint solutions in problems that meet these structural conditions. 
\end{abstract}

\section{Introduction} \label{section.1}

In combinatorial optimization problems, the objective is typically to identify a single optimal solution. However, this approach may be inadequate or impractical in real-world situations, where some constraints and factors are often overlooked or unknown in advance. This motivates the development of algorithms capable of finding multiple solutions, with \textit{diversity} playing a key role. A growing body of research has focused on finding diverse solutions in classical combinatorial problems, much of it emerging from the field of fixed-parameter tractability \cite{baste2019fpt, fomin2020diverse, hanaka2021finding, hanaka2022framework, misra2024parameterized, shida2024finding, drabik2024finding, kumabe2024max}. These studies show that finding diverse solutions is, in general, computationally more challenging than finding a single one. For instance, while \textsc{Matching} is solvable in polynomial time, finding two edge-disjoint matchings is NP-hard, even on 3-regular graphs \cite{fomin2020diverse}. 

Here, we aim to develop theoretically efficient algorithms that produce a collection of maximally diverse solutions. We use the sum of pairwise Hamming distances between solutions as our measure of diversity. In contrast with the aforementioned literature, we show that a broader class of diverse problems is computationally no harder than finding a single solution in polynomial time. Specifically, we generalize the polynomial-time solvability of $k$-\textsc{Diverse Minimum s-t Cuts} by De Berg et al. \cite{de2023finding} to a class of combinatorial problems whose solution sets form a distributive lattice. 

We state our main result in terms of a unified general problem: \textsc{Max-Sum $k$-Diverse Solutions}. Let $E$ be a finite set with $n$ elements, and let $\Gamma \subseteq 2^E$ be a set of feasible solutions. For two feasible solutions $X, Y \in \Gamma$, the symmetric difference, or Hamming distance, between them is defined as $X \triangle Y = (X \setminus Y) \cup (Y \setminus X)$. Let $(X_1, X_2, \ldots, X_k)$ be a collection of $k$ subsets of $E$. We consider the \textit{pairwise-sum} diversity measure:
\begin{equation*}
    d_\mathrm{sum}(X_1, X_2, \ldots, X_k) = \sum_{1 \leq i < j \leq k} |X_i \triangle X_j|.
\end{equation*}

Adopting the notation from Hanaka et al. \cite{hanaka2022framework}, we define \textsc{Max-Sum $k$-Diverse Solutions} as follows. Here, $k$ is a fixed constant; that is, $k$ is not part of the input. 

\begin{extthm}[\textsc{Max-Sum $k$-Diverse Solutions}] \label{problem:1}
Given a finite set $E$ of size $n$, an implicitly defined family $\Gamma$ of subsets of $E$, referred to as feasible solutions, and a membership oracle for $\Gamma \subseteq 2^E$, find a $k$-multiset $C = (X_1, X_2, \ldots, X_k)$ with $X_1, X_2, \ldots, X_k \in \Gamma$, such that $d_\mathrm{sum}(C)$ is maximum. 
\end{extthm}

Our main result is as follows. 

\begin{restatable}[]{theorem}{mainTheorem} \label{theorem.1}
    \textsc{Max-Sum $k$-Diverse Solutions} can be solved in polynomial time if the set of feasible solutions $\Gamma$ satisfies the following three properties: 
    \begin{enumerate}
        \item There is a relation $\leq$ such that the poset $(E, \leq)$ can be expressed as a disjoint union of chains and each feasible solution $X \subseteq E$ contains exactly one element from each chain. \label{property.1}
        \item The set of feasible solutions with componentwise order defines a distributive lattice. \label{property.2}
        \item A compact representation of this lattice can be constructed in polynomial time. \label{property.3}
        \end{enumerate}
\end{restatable}

Similarly to the approach of De Berg et al. \cite{de2023finding}, we achieve this result via a reduction to the \textit{submodular function minimization} problem (SFM) on a distributive lattice, which is known to be solvable in polynomial time \cite{grotschel2012geometric, iwata2001combinatorial, schrijver2000combinatorial}. More precisely, we show that the pairwise-sum measure (reformulated as a minimization objective) is a submodular function on a distributive lattice of appropriately ordered $k$-sized collections of feasible solutions. By applying this result, we obtain polynomial-time algorithms for finding maximally diverse $k$-sized collections of stable matchings, 
while also reproducing the findings of De Berg et al. for minimum $s$-$t$ cuts. 

For simplicity, we present our results in terms of the $d_\mathrm{sum}$ measure. However, in Section \ref{sec.5} we will show that the framework extends to at least two other measures of diversity: the \textit{coverage} ($d_\mathrm{cov}$) and \textit{absolute-difference} ($d_\mathrm{abs}$) measures. Lastly, we consider the problem of finding a largest set of pairwise disjoint solutions in problems whose feasible solution set satisfies properties \ref{property.1} and \ref{property.2} of the Theorem \ref{theorem.1}. In Section \ref{sec.6} we present an algorithm for this problem that bypasses the need for SFM. 

\paragraph{Roadmap.} In Section \ref{sec.2} we provide some preliminaries on lattice theory and submodular function minimization. Then, in Section \ref{sec.3} we present the proof of Theorem \ref{theorem.1}; that is, the reduction to SFM on a distributed lattice. Next, in Section \ref{sec.4}, we give examples of problems whose solution sets satisfy properties \ref{property.1}-\ref{property.3} of the theorem. In Section \ref{sec.5} we extend the result of Theorem \ref{theorem.1} to two other diversity measures. Then, in Section \ref{sec.6}, we present our oracle algorithm for finding a largest set of mutually disjoint solutions. We conclude in Section \ref{sec.7}. 

\section{Preliminaries} \label{sec.2}
In this section, we introduce the notation and some basic results used throughout the paper. For a more comprehensive discussion on sets and posets, we refer to \cite{harzheim2005ordered, stanley2011enumerative}, and for a detailed introduction to lattice theory, we refer to \cite{birkhoff1937rings, davey2002introduction, gratzer2009lattice}.

\paragraph{Sets, Multisets, and Tuples.} For $k \in \mathbb{N}$, we use $[k]$ to denote the set $\{1, \ldots, k\}$. The power set of a set $M$ is denoted by $2^M$. For any set $M$, we use the symbol $M^k$ for the cartesian product; $\{(a_1, a_2, \ldots, a_k) \; | \; a_i \in M\}$. The \textit{disjoint union} of two sets is simply their union, but with the additional information that the two sets have no elements in common. 

A \textit{multiset} is a set in which elements can appear multiple times. The number of times an element appears in a multiset is referred to as its \textit{multiplicity}. The \textit{sum} of two multisets $A$ and $B$, denoted by $A \uplus B$, is a multiset in which each element appears with a multiplicity equal to the sum of its multiplicity in $A$ and in $B$. We refer to a multiset of cardinality $k$ as a $k$-multiset. For a set $M$, we denote by $M_k$ a $k$-multiset where all elements are drawn from $M$. 

Unlike a multiset, where elements are unordered, a \textit{tuple} is a collection of possibly repeated elements that is ordered. A $k$-tuple is a tuple of $k$ elements. We denote a tuple by listing its elements within parenthesis and separated by commas; e.g., $(a, b, c, d)$. Note that the cartesian product of $k$ sets is a $k$-tuple. 

\paragraph{Posets.} A \textit{partially ordered set (poset)} $P = (X, \preceq_P)$ consists of a ground set $X$ along with a binary relation $\preceq_P$ on $X$ that satisfies reflexivity, antisymmetry, and transitivity. When the relation $\preceq_P$ is evident from the context, we often use the same notation for both the poset and its ground set. In case a poset is indexed by a subscript $i$, we use $\preceq_i$ to denote its order relation. 

The Hasse diagram $G(P)$ of $P$, is a directed graph where each element of $X$ is represented as a node, and an edge exists from element $a$ to element $b$ if $a \preceq_P b$ and no intermediate element $c$ satisfies $a \preceq_P c \preceq_P b$. Typically, vertices are arranged so that edge directions are implicitly understood as pointing upward. 

A poset $P^* = (X^*, \preceq_P^*)$ is called a \textit{subposet} of another poset $P = (X, \preceq_P)$ if (i) $X^* \subseteq X$ and (ii) for any $x, y \in X^*$ if $x \preceq_P^*y$ then $x \preceq_P y$. If the other direction of (ii) also holds, then we call $P^*$ the subposet of $P$ induced by $X^*$, and write $P^* = P[X^*]$. 

Given two posets $P = (X, \preceq_P)$ and $Q = (Y, \preceq_Q)$, their \textit{disjoint union} $P \sqcup Q$ is the disjoint union of $X$ and $Y$ together with relation $\preceq_{P+Q}$ where $x \preceq_{P+Q} y$ if one of the following conditions holds: (i) $x, y \in X$ and $x \preceq_P y$, or (ii) $x, y \in Q$ and $x \preceq_Q y$. Thus, the Hasse diagram of $P \sqcup Q$ consists of the disconnected Hasse diagrams of $P$ and $Q$ drawn together. 

A \textit{chain} is a subset of a poset in which every pair of elements is comparable, and an \textit{antichain} is a subset of a poset in which no two (distinct) elements are comparable. For any two elements $x$ and $y$ in a chain $E$ with order relation $\preceq_E$, we say that $x$ (resp. $y$) is a \textit{chain-predecessor} (\textit{chain-successor}) of $y$ if $x \preceq_E y$. A poset is called a \textit{chain decomposition} if the poset can be expressed as the disjoint union of chains. 

For a poset $P = (X, \preceq_P)$, an \textit{ideal} is a set $U \subseteq X$ where $u \in U$ implies that $v \in U$ for all $v \preceq_P u$. In terms of its Hasse diagram $G(P) = (X, E)$, a subset $U$ of $X$ is an ideal if and only if there is no incoming edge from $U$. We use $\mathcal{D}(P)$ to denote the family of all ideals of $P$. If $x \preceq_P y$ in the poset, then the closed interval from $x$ to $y$, denoted by $[x, y]$, is the poset with ground set $\{z \in X \; | \; x \preceq_P z \preceq_P y\}$ together with relation $\preceq_P$.

Now we introduce the notion of \textit{componentwise order}. Let $(X_i, \preceq_i)$, $i \in [r]$ be posets, with $r$ a positive integer, and let $Y \subseteq X_1 \times \dots \times X_r$. The componentwise order on $Y$ is an order relation $\preceq$ defined as follows: Given two tuples $(a_1, a_2, \ldots, a_r)$ and $(b_1, b_2, \ldots, b_r) \in Y$, we write $(a_1, a_2, \ldots, a_r) \preceq (b_1, b_2, \ldots, b_r)$ iff $a_i \preceq_i b_i$ for all $i \in [r]$. Note that we drop the subscript in $\preceq$ whenever the order relation is a component-wise order. If the posets $(X_i, \preceq_i)$, $i \in [r]$, are all the same poset $(X, \preceq)$, we use $\preceq^r$ to denote the componentwise order on $X^r$ and refer to it as the \textit{product order}. 

\paragraph{Lattices.} A \textit{lattice} is a poset $L = (X, \preceq)$ in which any two elements $x, y \in X$ have a (unique) greatest lower bound, or \textit{meet}, denoted by $x \wedge y$, as well as a (unique) least upper bound, or \textit{join}, denoted by $x \vee y$. We can uniquely identify $L$ by the tuple $(X, \vee, \wedge)$. The \textit{bottom}, or minimum, element in the lattice $L$ is denoted by $0_L := \bigwedge_{x \in L} x$. Likewise, the \textit{top}, or maximum, element of $L$ is given by $1_L := \bigvee_{x \in L} x$. 
A lattice $L'$ is a \textit{sublattice} of $L$ if $L' \subseteq L$ and $L'$ has the same meet and join operations as $L$. 
In this paper we only consider \textit{distributive lattices}, which are lattices whose meet and join operations satisfy distributivity; that is, $x \vee (y \wedge z) = (x \vee y) \wedge (x \vee z)$ and $x \wedge (y \vee z) = (x \wedge y) \vee (x \wedge z)$, for any $x,y,z \in L$. Note that a sublattice of a distributive lattice is also distributive. Every chain is a distributive lattice with $\max$ as join ($\vee$) and $\min$ as meet ($\wedge$).

Suppose we have a collection $L_1, \ldots, L_k$ of lattices $L_i = (X_i, \vee_i, \wedge_i)$ with $i \in [k]$. The \textit{(direct) product lattice} $L_1 \times \ldots \times L_k$ is a lattice with ground set $X = \{(x_1, \ldots, x_k) \, : \, x_i \in L_i\}$ and join $\vee$ and meet $\wedge$ operations acting component-wise; that is, $x \vee y = (x_1 \vee_1 y_1, \ldots, x_k \vee_k y_k)$ and $x \wedge y = (x_1 \wedge_1 y_1, \ldots, x_k \wedge_k y_k)$ for any $x, y \in X$. The lattice $L^k$ is the product lattice of $k$ copies of $L$, and is called the $k$-th \textit{power} of $L$. If $L$ is a distributive lattice, then $L^k$ is also distributive. 

A crucial notion we will need is that of \textit{join-irreducibles}. An element $x$ of a lattice $L$ is called \textit{join-irreducible} iff $x \neq 0_L$ and it cannot be expressed as the join of two elements $y, z \in L$ with $y, z \neq x$. In a lattice, any element is equal to the join of all join-irreducible elements lower than or equal to it. The set of join-irreducible elements of $L$ is denoted by $J(L)$. Note that $J(L)$ is a poset whose order is inherited from $L$. Due to Birkhoff's representation theorem---a fundamental tool for studying distributive lattices---every distributive lattice $L$ is isomorphic to the lattice $\mathcal{D}(J(L))$ of ideals of its poset of join-irreducibles, with union and intersection as join and meet operations. Hence, a distributive lattice $L$ can be uniquely recovered from its poset $J(L)$.

\begin{theorem}[Birkhoff's Representation Theorem \cite{birkhoff1937rings}] \label{thm.birkhoff}
Any distributive lattice $L$ can be represented as the poset of its join-irreducibles $J(L)$, with the order induced from $L$.
\end{theorem}

For a distributive lattice $L$, this implies that there exists a \textit{compact representation} of $L$ as the Hasse diagram $G(J(L))$ of its poset of join-irreducibles. This is useful when designing algorithms, as the size of $G(J(L))$ is $O(|J(L)|^2)$, while $L$ can have as many as $2^{|J(L)|}$ elements. Keep in mind, however, that Theorem \ref{thm.birkhoff} only guarantees the existence of such a compact representation; it does not provide a method to efficiently find the set $J(L)$.   

\begin{figure}
    \begin{center}
        \begin{tikzpicture}[scale=1.0, transform shape, baseline={(0,0)}]
    \tikzset{edge/.style = {->,> = latex}}
    
    \node[] (label) at (0, -0.75) {$L$};
    \node[] (n1) at (0, 0) {$x_1$};
    \node[] (n2) at (-1, 1) {\color{blue}{$x_2$}};
    \node[] (n3) at (1, 1) {\color{blue}{$x_3$}};
    \node[] (n4) at (0, 2) {$x_4$};
    \node[] (n5) at (2, 2) {\color{blue}{$x_5$}};
    \node[] (n6) at (1, 3) {$x_6$};
    
    \draw[] (n1) -- (n2);
    \draw[] (n1) -- (n3);
    \draw[] (n2) -- (n4);
    \draw[] (n3) -- (n4);
    \draw[] (n3) -- (n5);
    \draw[] (n4) -- (n6);
    \draw[] (n5) -- (n6);
    \end{tikzpicture}
    \hspace{1.25cm}
    \begin{tikzpicture}[scale=1.0, transform shape, baseline={(0,0)}]
    \tikzset{edge/.style = {->,> = latex}}
    
    \node[] (label) at (0, -0.75) {$\mathcal{D}(J(L))$};
    \node[] (n1) at (0, 0) {$\emptyset$};
    \node[] (n2) at (-1, 1) {\color{blue}{$\{x_2\}$}};
    \node[] (n3) at (1, 1) {\color{blue}{$\{x_3\}$}};
    \node[] (n4) at (0, 2) {$\{x_2, x_3\}$};
    \node[] (n5) at (2, 2) {\color{blue}{$\{x_3, x_5\}$}};
    \node[] (n6) at (1, 3) {$\{x_2, x_3, x_5\}$};
    
    \draw[] (n1) -- (n2);
    \draw[] (n1) -- (n3);
    \draw[] (n2) -- (n4);
    \draw[] (n3) -- (n4);
    \draw[] (n3) -- (n5);
    \draw[] (n4) -- (n6);
    \draw[] (n5) -- (n6);
    \end{tikzpicture}
    \hspace{1.25cm}
    \begin{tikzpicture}[scale=1.0, transform shape, baseline={(0,0)}]
    \tikzset{edge/.style = {->,> = latex}}
    
    \node[] (label) at (-0.75, -0.75) {$G(J(L))$};
    \node[] (n5) at (0, 2) {\color{black}{$x_5$}};
    \node[] (n2) at (-1.5, 0.5) {\color{black}{$x_2$}};
    \node[] (n3) at (0, 0.5) {\color{black}{$x_3$}};
    
    \draw[] (n5) -- (n3);
    \end{tikzpicture}
    \end{center}
    \vspace*{-4mm}
    \caption{Example of Birkhoff's representation theorem for distributive lattices. The left is a distributive lattice $L$, the middle is the isomorphic lattice $\mathcal{D}(J(L))$ of ideals of join-irreducibles of $L$, and the right shows the compact representation $G(J(L))$ of $L$. The join irreducible elements of $L$ and $\mathcal{D}(J(L))$ are highlighted in blue.}
    \label{fig:compactRepresentation}
\end{figure}
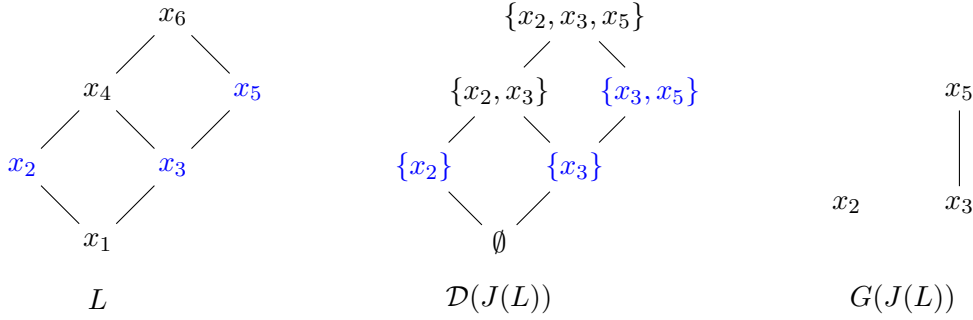

\paragraph{Submodular Function Minimization.} Let $f:X \rightarrow \mathbb{R}$ be a real-valued function on a lattice $L = (X, \preceq)$. We say that $f$ is \textit{submodular} on $L$ if 
\begin{equation} \label{eq:submodular}
    f(x \wedge y) + f(x \vee y) \leq f(x) + f(y), \quad \text{for all } x,y \in X\text{.}
\end{equation}
If $-f$ is submodular on $L$, then we say that $f$ is \textit{supermodular}. If $f$ satisfies Eq. \eqref{eq:submodular} with equality we say that $f$ is \textit{modular} in $L$. 
The \textit{submodular function minimization} problem (SFM) on lattices is, given a submodular function $f$ on $L$, to find an element $x \in L$ such that $f(x)$ is minimum. An important fact that we use in our work is that the sum of submodular functions is also submodular. Also, note that minimizing $f$ is equivalent to maximizing $-f$. 

Consider the special case of a lattice whose ground set $X \subseteq 2^U$ is a family of subsets of a set $U$, and meet and join are intersection and union of sets, respectively. It is known that any function $f$ satisfying \eqref{eq:submodular} on such a lattice can be minimized in polynomial time in $|U|$ \cite{grotschel2012geometric, iwata2001combinatorial, schrijver2000combinatorial}. This holds when assuming that for any $Y \subseteq U$, the value of $f(Y)$ is given by an \textit{evaluation oracle} that also runs in polynomial time in $|U|$. The current fastest algorithm for SFM on sets runs in $O(|U|^3 T_\mathrm{EO})$ time \cite{jiang2021minimizing}, where $T_\mathrm{EO}$ is the time required for one call to the evaluation oracle. 

Due to Birkhoff's theorem, the seemingly more general case of SFM on distributive lattices can be reduced to SFM on sets (see e.g. \cite[Sec. 3.1]{bolandnazar2015note} for details). Hence, any polynomial-time algorithm for SFM on sets can be used to minimize a submodular function $f$ defined on a distributive lattice $L$. An important remark is that the running time depends on the size of the set $J(L)$ of join-irreducibles.

\begin{theorem}[{\cite[Note 10.15]{murota2003} and \cite[Thm.1]{markowsky2001overview}}] \label{theorem:sfmDistributiveLattices}
For any distributive lattice $L$, given by its poset of join-irreducibles $J(L)$, a submodular function $f: L \rightarrow \mathbb{R}$ can be minimized in polynomial time in $|J(L)|$, provided a polynomial time evaluation oracle for $f$.
\end{theorem}

\section{The Reduction to SFM} \label{sec.3}

In this section, we prove Theorem \ref{theorem.1} by reducing \textsc{Max-Sum $k$-Diverse Solutions} to SFM on a distributive lattice, under the assumption that the feasible solution set satisfies properties \ref{property.1}–\ref{property.3} of the theorem. Our approach closely follows the proof by De Berg et al. \cite{de2023finding} for minimum $s$-$t$ cuts, with modifications to generalize the results and streamline certain arguments. 

The proof is divided into four parts, each supported by a corresponding lemma. The Distributivity Lemma (Lemma \ref{lemma.1}) shows that the set of \textit{left-right ordered} $k$-tuples of feasible solutions, with product order, defines a distributive lattice $L^*$. The Cost-Equivalence Lemma (Lemma \ref{lemma.2}) further shows that optimizing the diversity over this lattice is the same as optimizing over the original set $\Gamma_k$ of $k$-multisets of $\Gamma$. Hence, we can restrict ourselves to the elements of $L^*$. Next, the Submodularity Lemma (Lemma \ref{lemma.3}) establishes that the pairwise-sum measure (reformulated as a minimization objective) is a submodular function on $L^*$. Finally, the Compactness Lemma (Lemma \ref{lemma.4}) ensures that a compact representation of $L^*$ can be constructed in polynomial time.  

We begin by establishing some consequences of properties \ref{property.1}–\ref{property.3} of Theorem \ref{theorem.1}. Consider a ground set $E$ and a set of feasible solutions $\Gamma \subseteq 2^E$ for which the properties hold. By property \ref{property.1}, we know that there is a poset associated to $E$ that is the disjoint union of $r$ chains $(E_i, \preceq_i)$, $i \in [r]$, and that each feasible solution $X \in \Gamma$ contains exactly one element from each chain, meaning $\Gamma \subseteq E_1 \times \dots \times E_r$. Then, the set $\Gamma$, with componentwise order $\preceq$, forms a poset of feasible solutions $L = (\Gamma, \preceq)$. Furthermore, by properties \ref{property.1} and \ref{property.2}, this poset is a distributive lattice, with join ($\vee$) and meet ($\wedge$) given by componentwise maximum and minimum. Let us now consider the poset $(\Gamma^k, \preceq^k)$ of $k$-tuples of feasible solutions, with product order $\preceq^k$.  We say that a $k$-tuple $C = (X_1, X_2, \ldots, X_k)$ in $\Gamma^k$ is in \textit{left-right order} if $X_i \preceq X_j$ for all $i < j$. That is, the feasible solutions in $C$ are arranged in non-decreasing order according to relation $\preceq$. Let $\Gamma_{\mathrm{lr}}^k \subseteq \Gamma^k$ denote the subset of left-right ordered $k$-tuples.

\paragraph{Part 1: Distributivity.} 
We now establish the first of the four lemmas. 

\begin{lemma}[Distributivity Lemma] \label{lemma.1}
    The poset $L^* = (\Gamma_{\mathrm{lr}}^k, \preceq^k)$ is a distributive lattice. 
\end{lemma}
\begin{proof} 
    By property \ref{property.2} of Theorem \ref{theorem.1}, $L = (\Gamma, \preceq)$ is a distributive lattice. 
    Now, let $L^k = (\Gamma^k, \preceq^k)$ be the $k$-th power of $L$. We know that the product of distributive lattices is distributive; hence $L^k$ is a distributive lattice. Moreover, since $\Gamma_{\mathrm{lr}}^k \subseteq \Gamma^k$, the poset $L^*$ is a sublattice of $L^k$. As any sublattice of a distributive lattice is itself distributive, the lemma follows. 
\end{proof}

\paragraph{Part 2: Cost equivalence.} 
Following the proof of the Distributivity Lemma, we now establish an equivalence between the costs of maximum diversity solutions in the sets $\Gamma^k$ and $\Gamma_{\mathrm{lr}}^k$. 
(Note that this is the same as establishing the equivalence between the sets $\Gamma_k$ and $\Gamma_{\mathrm{lr}}^k$, since a $k$-multiset over $\Gamma$ has the same diversity as each of its up to $n!$ permutations---each a $k$-tuple---in $\Gamma^k$.) 
For this, we use the notion of element multiplicity. Let $C \in \Gamma^k$ be a $k$-tuple of solutions. The multiplicity $\mu_e(C)$ of an element $e \in E$, with respect to $C$, is the number of feasible solutions in $C$ that contain $e$. Since a feasible solution contains no repeated elements, $\mu_e(C)$ is also the number of times $e$ appears in the multiset sum of the solutions in $C$. An immediate consequence of property \ref{property.1} of Theorem \ref{theorem.1} is the following.

\begin{observation} \label{obs.1}
    For any two $X, Y \in L$, we have $X \uplus Y = (X \wedge Y) \uplus (X \vee Y)$.
\end{observation}
\begin{proof}
    Let $X = (x_1, \ldots, x_r)$ and $Y = (y_1, \ldots, y_r)$, with $x_i, y_i \in E_i$ and $i \in [r]$, where $E_i$ is the $i$-th chain in the chain decomposition of $E$. By property \ref{property.1}, we know that the meet and join of two elements $L$ are given by componentwise minimum and maximum. That is, \[X \wedge Y = (\min(x_1, y_1), \ldots, \min(x_r, y_r)), \quad \text{and} \quad
    X \vee Y = (\max(x_1, y_1), \ldots, \max(x_r, y_r)).\]
    Hence, if $x_i = y_i$, the element $x_i$ appears twice in the multiset sum $X \uplus Y$ and twice in the sum $(X \wedge Y) \uplus (X \vee Y)$. If $x_i \neq y_i$, then $x_i$ appears in either the join or the meet of $X$ and $Y$, and similarly for element $y_i$. Finally, if an element $e \in E$ is not in $X \cup Y$, then it is neither the minimum nor maximum of any entry and therefore cannot appear in $(X \wedge Y) \cup (X \vee Y)$. Since this holds for each $i \in [k]$, the observation is proven. 
\end{proof}

Observation \ref{obs.1} implies that the join and meet operations of the lattice $L$ of feasible solutions preserve element multiplicities. This is relevant in the proof of Lemma \ref{lemma.2}. 
Consequently, any $k$-tuple in $\Gamma^k$ can be reordered into a left-right ordered form while preserving element multiplicities, as stated in the following claim. 

\begin{claim} \label{claim.0}
    For every $k$-tuple $C \in \Gamma^k$ there exists a left-right ordered $k$-tuple $\hat{C} \in L^k_\mathrm{lr}$ such that $\mu_e(C) = \mu_e(\hat{C})$ for all $e \in E$.
\end{claim}
\begin{proof}
    To prove this, we give an algorithm that takes any $k$-tuple $C \in \Gamma$ of feasible solutions and transforms it into a $k$-tuple $\hat{C} \in \Gamma^k_\mathrm{lr}$ that is in left-right order while preserving element multiplicities. The algorithm iteratively rearranges the elements using join ($\vee$) and meet ($\wedge$) operations to enforce the desired ordering. We remark that any sorting algorithm that uses join and meet as comparison operators would achieve the same result. The algorithm is outlined below. 
    
\begin{algorithm}[H]
\caption[Caption for LOF]{LRO($C = (X_1, \ldots, X_k)$)} \label{algo:leftRightOrder}
\vspace{0.5em}
{\setlist{nolistsep}
    \begin{enumerate}
    \item For each $i \in \{1, \ldots, k-1\}$:
        \begin{enumerate}
        \item For each $j \in \{i+1, \ldots, k\}$:
            \begin{enumerate}
                \item $A \leftarrow X_i \wedge X_j$ and $B \leftarrow X_i \vee X_j$
                \item $X_i \leftarrow A$ and $X_j \leftarrow B$
        \end{enumerate}
    \end{enumerate}
    \item Return $\hat{C} \leftarrow  (X_1, X_2, \ldots, X_k)$.
    \end{enumerate}
}
\vspace{0.5em}
\end{algorithm}

The algorithm iterates through the tuple, comparing each element $X_i$ with every subsequent element $X_j$ for $j > i$. In each comparison, $X_i$ is replaced by their meet $X_i \wedge X_j$ and $X_j$ is replaced by their join $X_i \vee X_j$. By the definition of join and meet, at the end of iteration $i$, the solution $X_i$ is a predecessor of all other solutions $X_j \in C$ with $j > i$. By repeating this process over all pairs, the elements are gradually rearranged so that $X_i \preceq X_j$ for all $i < j$. 

It is clear that by the end of the algorithm, the $k$-tuple $\hat{C}$ is in left-right order. It only remains to verify that $\mu_e(C) = \mu_e(\hat{C})$ for each $e \in E$. By Observation \ref{obs.1} the multiplicity of elements is preserved under the join and meet operations. Hence the multiplicity of elements remains invariant at each pairwise comparison operation between solutions $X_i$ and $X_j$, and the claim is proven.     
\end{proof}

Now, consider the pairwise-sum diversity measure introduced in Section \ref{section.1}. We can rewrite it directly in terms of the multiplicity as 
\begin{align}
    & d_{\text{sum}}(C) = 2 \left[r \binom{k}{2} - \sum_{e \in E} \binom{\mu_e(C)}{2} \right], \label{eq:dsumMultiplicity}
\end{align}
for any $C \in \Gamma^k$. Notice that the terms outside the summation are constant. (Recall that $\binom{m}{n} = 0$ when $n > m$.) This formulation highlights that maximizing $d_\mathrm{sum}$ depends only on the distribution of elements across feasible solutions, rather than their specific ordering within a tuple $C$. The following lemma is an immediate consequence of Claim \ref{claim.0} and Equation \eqref{eq:dsumMultiplicity}. 

\begin{lemma}[Cost-Equivalence Lemma] \label{lemma.2}
    For any $C \in \Gamma^k$ there exists $\hat{C} \in \Gamma_{\mathrm{lr}}^k$ such that $d_{\mathrm{sum}}(\hat{C}) = d_{\mathrm{sum}}(C)$. 
\end{lemma}
\begin{proof}
By Claim \ref{claim.0}, for any $C \in \Gamma^k$ there exists a tuple $\hat{C} \in \Gamma^k_\mathrm{lr}$ with the same element multiplicities. Equation \eqref{eq:dsumMultiplicity} establishes that the diversity $d_\mathrm{sum}$ of a collection of solutions depends solely on element multiplicities. Then, $C$ and $\hat{C}$ must also have the same diversity value. 
\end{proof}

Lemma \ref{lemma.2} tells us that in order to solve \textsc{Max-Sum $k$-Diverse Solutions} we do not need to optimize over the set of $k$-element multisets of $\Gamma$. Instead, we can optimize over the set $\Gamma_{\mathrm{lr}}^k$ of $k$-tuples that are in left-right order. Moreover, it follows from Eq. \eqref{eq:dsumMultiplicity} that maximizing $d_\mathrm{sum}$ is equivalent to minimizing 
\begin{align}
    & \hat{d}_\mathrm{sum}(C) = \sum_{e \in E} \binom{\mu_e(C)}{2}.
    \label{eq:dsumMultiplicityMin} 
\end{align}

Hence, solving \textsc{Max-Sum $k$-Diverse Solutions} reduces to minimizing $\hat{d}_\mathrm{sum}(C)$ in the lattice $L^*$. All we have left to do to complete the reduction to SFM is show that $\hat{d}_\mathrm{sum}(C)$ is submodular in the lattice $L^*$. 

\paragraph{Part 3: Submodularity.} 
We begin with three claims regarding the multiplicity function $\mu_e(C)$ on $L^*$. We remark that these claims rely crucially in property \ref{property.1} of Theorem \ref{theorem.1}.

\begin{claim} \label{claim.1}
The multiplicity function $\mu_e: \Gamma_{\mathrm{lr}}^k \rightarrow \mathbb{N}$ is modular on $L^*$.  
\end{claim}
\begin{proof}
Consider $C_1 = (X_1, \ldots, X_k)$ and $C_2 = (Y_1, \ldots, Y_k) \in \Gamma^k_\mathrm{lr}$. By the definition of product order, we know that $C_1 \vee C_2 = (X_1 \vee Y_1, \ldots, X_k \vee Y_k)$ and $C_1 \wedge C_2 = (X_1 \wedge Y_1, \ldots, X_k \wedge Y_k)$. Now, by Observation \ref{obs.1}, we have $X_i + Y_i = (X_i \vee Y_i) + (X_i \wedge Y_i)$ for all $i \in [k]$. Taking the multiset sum over each $i$ and rearranging, we have $(X_1 + \ldots + X_k) + (Y_1 + \ldots + Y_k) = ((X_1 \vee Y_1) + \ldots + (X_k \vee Y_k)) + ((X_1 \wedge Y_1) + \ldots + (X_k \wedge Y_k))$. That is, $\mu_e(C_1) + \mu_e(C_2) = \mu_e(C_1 \vee C_2) + \mu_e(C_1 \wedge C_2)$ for each element $e \in E$. By definition of modular function, the claim follows. 
\end{proof}

Notice that Claim \ref{claim.1} holds for $L^k$ in general, not just in the lattice of left-right ordered $k$-tuples. The following two claims however, are only true for $L^*$. We use $E(C)$ to denote the set of elements $\bigcup_{X \in C} X$ for a tuple $C \in \Gamma^k$.

\begin{claim} \label{claim.2}
For any $C = (X_1, \ldots, X_k)$ in $L^*$, the element $e \in E(C)$ appears in every feasible solution of a contiguous subsequence $C' = (X_i, \ldots, X_j)$ of $C$, $1 \leq i \leq j \leq k$, with size $|C'| = \mu_e(C)$.
\end{claim}
\begin{proof}
    The case when $\mu_e(C) = 1$ is trivial. Next, we prove the case when $\mu_e(C) \geq 2$. By contradiction, suppose that $e$ does not appear in a contiguous subsequence of $C$. Then, there exists some feasible solution $X_h \in C$ with $i < h < j$ such that $e \in X_i$, $e \not\in X_h$, and $e \in X_j$. We know that the tuple $C$ is in left-right order, thus we have that $X_i \preceq X_j$ for every $i < j$. Without loss of generality, let $e$ belong to the $\ell$-th chain of the chain decomposition of $E$. 
    Since $e \in X_i$, there must be an element $f \in X_h$ such that $e \preceq_\ell f$. But since $e \in X_j$, it must also hold that $f \preceq_\ell e$, which can only be satisfied if $e = f$. However, we assumed that $e \not\in X_h$, hence $e \neq f$, which gives the necessary contradiction. 
\end{proof}

Claim \ref{claim.2} enables us to represent the containment of an element $e$ in a tuple $C \in L^*$ as the interval $I_e(C) = (i, j)$, where $i \leq j$, of length $\mu_e(C)$ defined on the set of integers $[k]$. In this interval representation, the elements of $I_e(C)$ correspond bijectively to the positions taken by the solutions that contain the element $e$ in the tuple $C$. See Figure \ref{fig:intervalRepresentation} for an example. This representation is useful for proving the subsequent claim. 

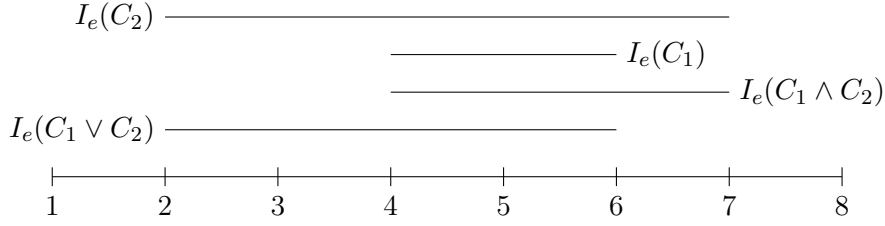
\begin{figure}
    \centering
    \begin{tikzpicture}[scale=0.5]
    \draw (3,0) -- (24,0); 
    \foreach \x in {1,2,...,8} {
        \draw (\x*3,0.25) -- (\x*3,-0.25) node[below] {\x};
    }
    \draw (6,1.25) node[left] {$I_e(C_1 \vee C_2)$} -- (18,1.25); 
    \draw (12,2.25) -- (21,2.25) node[right] {$I_e(C_1 \wedge C_2)$}; 
    \draw (12,3.25) -- (18,3.25) node[right] {$I_e(C_1)$}; 
    \draw (6,4.25) node[left] {$I_e(C_2)$} -- (21,4.25); 
    
    \end{tikzpicture}
    \caption{Interval representation of the containment of an element $e \in E$ in left-right ordered collections $C_1$ and $C_2$ of feasible solutions, as well as in their join and meet. In the example, there are $k = 8$ solutions in each tuple, and eight corresponding elements in the domain of the intervals. Observe that neither $I_e(C_1 \vee C_2)$ nor $I_e(C_1 \wedge C_2)$ are longer than $I_e(C_1)$ or $I_e(C_2)$. Also, note that the corresponding sums of their lengths are equal. 
    }
    \label{fig:intervalRepresentation}
\end{figure}

\begin{restatable}[]{claim}{multiplicityProperty} \label{claim.3}
For any two $C_1, C_2 \in L^*$ and $e \in E(C_1) \cup E(C_2)$, it holds that $\max(\mu_e(C_1 \vee C_2), \mu_e(C_1 \wedge C_2)) \leq \max(\mu_e(C_1), \mu_e(C_2))$.
\end{restatable}
\begin{proof}
    We prove this by case distinction on the containment of $e$ in $E(C_1) \cup E(C_2)$. Without loss of generality, assume that $e$ belongs to the $\ell$-th chain in the chain decomposition of $E$ dictated by property \ref{property.1} of Theorem \ref{theorem.1}; i.e., $e \in E_\ell$. There are three cases to consider: $e \in E(C_1) \setminus E(C_2)$, $e \in E(C_2) \setminus E(C_1)$, and $e \in E(C_1) \cap E(C_2)$. 
\begin{description}
   \item[Case 1: $e \in E(C_1) \setminus E(C_2)$.] 
    By Claim \ref{claim.1}, we know that $\mu_e(C_1 \vee C_2) + \mu_e(C_1 \wedge C_2) = \mu_e(C_1)$. Since $\mu_e(C_1 \vee C_2)$ and $\mu_e(C_1 \wedge C_2)$ are nonnegative, neither term can be greater than $\mu_e(C_1)$ and thus, the claim holds. 
   \item[Case 2: $e \in E(C_2) \setminus E(C_1)$.] 
   This case is symmetrical to Case 1. 
   \item[Case 3: $e \in E(C_1) \cap E(C_2)$.]
   Consider the interval representation of edge $e$ in $E(C_1)$ and $E(C_2)$. Let $I_e(C_1) = (\alpha, \beta)$ and $I_e(C_2) = (\sigma, \tau)$ be such intervals. There are two subcases to consider: $I_e(C_1) \cap I_e(C_2) = \emptyset$, and $I_e(C_1) \cap I_e(C_2) \neq \emptyset$. 
   \begin{description}
       \item[Subcase 3.1: $I_e(C_1) \cap I_e(C_2) = \emptyset$.] We claim that $\max({\mu_e(C_1 \vee C_2), \mu_e(C_1 \wedge C_2)})$ is equal to $\max(\mu_e(C_1), \mu_e(C_2))$ in this subcase. To see this, w.l.o.g., suppose that $\beta < \sigma$. Then, because $C_2$ is in left-right order, the solutions of $C_2$ in the interval $(\alpha, \beta)$ each contain a chain-predecessor of $e$. Then, by definition of the join operation in $L^*$, we have $I_e(C_1 \vee C_2) = (\alpha, \beta)$. Similarly, the solutions of $C_1$ in the interval $(\sigma, \tau)$ each contain a chain-successor of $e$. Hence, by the meet operation in $L^*$, we have $I_e(C_1 \wedge C_2) = (\sigma, \tau)$. Comparing the length of the intervals, we obtain that $\mu_e(C_1 \vee C_2) = \mu_e(C_1)$ and $\mu_e(C_1 \wedge C_2) = \mu_e(C_2)$, from which the claim follows. 
       \item[Subcase 3.2: $I_e(C_1) \cap I_e(C_2) \neq \emptyset$.] We have two further subcases to consider: $I_e(C_1) \not\subseteq I_e(C_2)$ (or $I_e(C_2) \not\subseteq I_e(C_1)$), and $I_e(C_2) \subseteq I_e(C_1)$ (or vice versa). 
       \begin{description}
           \item[Subcase 3.2.1: $I_e(C_1) \not\subseteq I_e(C_2)$.]
           The proof of this subcase is analogous to the proof of subcase (3.1), where we also obtain that $\max({\mu_e(C_1 \vee C_2), \mu_e(C_1 \wedge C_2)})$ equals $\max(\mu_e(C_1), \mu_e(C_2))$. It suffices to notice that the intersection of intervals contributes equally to the multiplicities $\mu_e(C_1)$ and $\mu_e(C_2)$, as well as to $\mu_e(C_1\vee C_2)$ and $\mu_e(C_1 \wedge C_2)$. Hence we can treat the complement in the same way as Subcase (3.1).
           \item[Subcase 3.2.2: $I_e(C_2) \subseteq I_e(C_1)$.] Recall that $I_e(C_1) = (\alpha, \beta)$, and $I_e(C_2) = (\sigma, \tau)$. Then, in this subcase: $\alpha \leq \sigma \leq \tau \leq \beta$. Again, by definition of join and meet, we have that $I_e(C_1 \vee C_2) = (\alpha, \tau)$ and $I_e(C_1 \wedge C_2) = (\sigma, \beta)$. Now, since $\tau - \alpha \leq \beta - \alpha$ and $\beta - \sigma \leq \beta - \alpha$, we obtain $\max(\mu_e(C_1 \vee C_2), \mu_e(C_1 \wedge C_2)) \leq \max(\mu_e(C_1), \mu_e(C_2))$, which is what we wanted.
       \end{description}
   \end{description}
\end{description}

Since the claim is true for all cases covered and all cases have been considered, the result is proved.
\end{proof}

With these three claims at our disposal, we are just one step away from proving the submodularity of $\hat{d}_\mathrm{sum}$. Let $B_e: U_{\mathrm{lr}}^k \rightarrow \mathbb{N}$ be the function defined by $B_e(C) = \binom{\mu_e(C)}{2}$. We can rewrite Eq. \eqref{eq:dsumMultiplicityMin} as $\hat{d}_\mathrm{sum}(C) = \sum_{e \in E(C)} B_e(C)$. The following proposition is a consequence of Claims \ref{claim.1} and \ref{claim.3}. 

\begin{proposition} \label{prop.1}
For any two $C_1, C_2 \in L^*$ and $e \in E$, we have $B_e(C_1 \vee C_2) + B_e(C_1 \wedge C_2) \leq B_e(C_1) + B_e(C_2)$.
\end{proposition}
\begin{proof}

We show that $\binom{a}{2} + \binom{b}{2} \leq \binom{c}{2} + \binom{d}{2}$ for $a, b, c, d \in \mathbb{N}$, given the following conditions: (i) $a + b = c + d$, and (ii) $\max(a, b) \leq \max(c, d)$. By setting $a = \mu_e(C_1 \vee C_2)$, $b = \mu_e(C_1 \wedge C_2)$, $c = \mu_e(C_1)$, and $d = \mu_e(C_2)$, the proposition follows, as conditions (i) and (ii) correspond to the properties of the multiplicity function stated in Claims \ref{claim.1} and \ref{claim.2}, respectively.

Without loss of generality, assume that $a \leq b$ and $c \leq d$. Condition (ii) ensures that $b \leq d$, and together with (i) implies that $a \geq c$, and $b \geq c$. Therefore, $(b - a) \leq (d - c)$. It is folklore that for $a, b > 0$ with fixed sum, the product $a \cdot b$ becomes larger as $|a - b|$ becomes smaller. Hence, we have $a \cdot b \geq c \cdot d$. Next, from condition (i), we know that $(a+b)^2 = (c+d)^2$. Combining this with the inequality just established, and subtracting $(a+b)$ and $(c+d)$ from each side, respectively, we get $a(a-1) + b(b-1) \leq c(c-1) + d(d-1)$. By definition of the binomial coefficient, the proposition immediately follows. 
\end{proof}

With Proposition \ref{prop.1}, we are ready to prove that $\hat{d}_\mathrm{sum}(C)$ is submodular in the lattice $L^*$. 

\begin{lemma}[Submodularity Lemma] \label{lemma.3}
    The function $\hat{d}_{\mathrm{sum}}: \Gamma_{\mathrm{lr}}^k \rightarrow \mathbb{N}$ is submodular in $L^*$.
\end{lemma}
\begin{proof}
    Proposition \ref{prop.1} states that the function $B_e(C)$ is submodular in the lattice $L^*$. Now, recall that the sum of submodular functions is also submodular. Then, taking the sum of $B_e(C)$ over all elements $e \in E$ results in a submodular function. From here, notice that $B_e(C) = 0$ when $\mu_e(C) < 2$. This means that such elements do not contribute to the sum. It follows that, for any two $C_1, C_2 \in L^*$, we have
\begin{equation*}
    \sum_{e \in E}B_e(C_1 \vee C_2) + \sum_{e \in E}B_e(C_1 \wedge C_2) \leq \sum_{e \in E}B_e(C_1) + \sum_{e \in E}B_e(C_2). 
\end{equation*}
Each sum in this inequality corresponds to the definition of $\hat{d}_{sum}$ applied to the arguments $C_1 \vee C_2$, $C_1 \wedge C_2$, $C_1$ and $C_2$, respectively. Hence, by definition of submodularity, we obtain our desired result.
\end{proof}

\paragraph{Part 4: Compactness.} 
While Lemmas \ref{lemma.1}-\ref{lemma.3} already demonstrate the reduction of \textsc{Max-Sum $k$-Diverse Solutions} to SFM, this reduction alone does not guarantee an efficient algorithm. To complete the proof of Theorem \ref{theorem.1}, it remains to show that a compact representation of the left-right ordered lattice $L^*$ exists and can be constructed efficiently. This is done in the following lemma, assuming that property \ref{property.3} holds. 

\begin{lemma}[Compactness Lemma] \label{lemma.4}
    A compact representation of $L^*$ can be constructed in polynomial time in $n$ and $k$. 
\end{lemma}
\begin{proof}
    By Birkhoff's representation theorem, we need only specify the set of join-irreducibles of $L^*$ for us to obtain a compact representation in $O(|J(L^*)|^2)$ time.    
    Next, we prove that the set of join-irreducibles of $L^*$ is of size $O(kn)$ and is given by
    \begin{center}  
        $J(L^*) = \bigcup_{i = 1}^k J_i$, where $J_i := \{(\underbrace{0_L, \ldots, 0_L}_{i-1 \text{ times}}, \underbrace{p, \ldots, p}_{k-i+1 \text{ times}})\, : \, p \in J(L)\}$;
    \end{center} 
that is, the join-ireducibles of $L^*$ are determined by the irreducibles of $L$. Since, by property \ref{property.3} of Theorem \ref{theorem.1}, the join-irreducibles of $L$ can be found in polynomial time, the lemma follows.   

    First, recall that an element $x \in L^*$, where $x \neq 0_L$, is a join-irreducible if and only if $x$ has a single immediate predecessor. To prove our claim, we show that (i) the $k$-tuples $J_i$, with $1 \leq i \leq k$, are in $L^*$ and satisfy this property, and (ii) that no other tuple in $L^*$ satisfies it. 
    
    To prove (i), let $C(i, p)$ denote the $k$-tuple $(0_L, \ldots, 0_L, p, \ldots, p) \in J_i$, where the first $i-1$ entries contain $0_L$ and the remaining $\ell = k-i+1$ entries contain the element $p$, with $i \in [k]$ and $p \in J(L)$. It is clear that $C(i, p) \in L^*$, since each entry in $C(i, p)$ is a feasible solution from the set $\Gamma$, and $0_L \preceq X$ for any $X \in \Gamma$. Notice that any predecessor of $C(i, p)$ in $L^*$ must be of the form $(0_L, \ldots, 0_L, q_\ell, \ldots, q_k)$ with $q_j \in L$, where $q_j \preceq p$ for all $j \in [\ell, k]$ and $(q_\ell, \ldots, q_{k})$ is in left-right order. Thus, an immediate predecessor of $C(i, p)$ in $L^*$ must be a $k$-tuple of the form $Q(i, q, p) := (0_L, \ldots, 0_L, q, p, \ldots, p)$ constructed by replacing the $i$th entry of $C(i, p)$ with an element $q \in L$ such that $q \prec p$. Assume, for the sake of contradiction, that $C(i, p)$ has more than a single immediate predecessor. Then, the join of any two such tuples $Q(i, x, p)$ and $Q(i, y, p)$ would result in $C(i, p)$. By construction, this can only happen if either $x = p$ or $y = p$, but we assumed that $x \prec p$ and $y \prec p$, which is a contradiction. Hence, $C(i, p)$ must have a single immediate predecessor. Since this holds for all $i\in [k]$ and arbitrary $p$, it follows that each of the tuples in $J(L^*)$ has a single immediate predecessor. 

    It remains to show (ii); that is, that there is no tuple in $L^* \setminus \bigcup_{i = 1}^k J_i$ which is also a join-irreducible of $L^*$. For this, it suffices to show that any tuple in $L^* \setminus \bigcup_{i = 1}^k J_i$ can be written as the join of two other elements in $L^*$. Indeed, consider an arbitrary tuple $T \in L^* \setminus \bigcup_{i = 1}^k J_i$. Such a tuple has the form $T = (0_L, \ldots, 0_L, a, \ldots, b, \ldots)$, with $a \prec b$, where $a \in L$ is the first entry in $T$ that is not the bottom element $0_L$, and $b \in L$ is the first entry in $T$ that is neither $0_L$ nor $a$. Now, consider the tuple $T_1 = (0_L, \ldots, 0_L, a, \ldots, a)$ obtained by replacing every successor of $a$ in $T$ by element $a$, and the tuple $T_2 = (0_L, \ldots, 0_L, b, \ldots)$ obtained from $T$ by replacing every predecessor of $b$ by $0_L$. It is clear that $T_1, T_2 \in L^*$. Moreover, it holds that $T_1 \prec T$ and $T_2 \prec T$, while $T_1$ and $T_2$ are incomparable. Then, by definition of the join operation in the lattice $L^*$, we have that $T = T_1 \vee T_2$. Hence, $T$ can be can be written as the join of two other elements in $L^*$ and thus, is not a join-irreducible of $L^*$. 
    
    From (i) and (ii) above, we have thus shown that the set of join-irreducibles $J(L^*)$ is given by $\bigcup_{i = 1}^k J_i$. To conclude the proof, we look at the size of $J(L^*)$. First, observe that the index $i$ runs from $1$ to $k$. Also, by property \ref{property.3} of Theorem \ref{theorem.1} we know that $|J(L)| = O(n)$. It then follows that $|J(L^*)| = O(kn)$.
\end{proof}

\begin{remark}
    Consider the lattice $L$ of feasible solutions. Let $J(e) \in J(L)$ be the minimum join-irreducible element of $L$ that contains element $e \in E$. It is clear that every join-irreducible element of $L$ is of this form. Hence, we may write $J(L) = \{J(e) \; | \; e \in E\}$. Then, to construct $J(L)$, it is sufficient to give a procedure to find $J(e)$ for a given $e \in E$. We could then rewrite Property \ref{property.3} of Theorem \ref{theorem.1} by assuming the existence of such an algorithm running in polynomial time. We remark that such a procedure has been shown to exist for various combinatorial objects like stable matchings and the set of \textit{consistent cuts of a computation} \cite{garg2001slicing, mittal2001computation, garg2018applying}.    
\end{remark}

With Lemma \ref{lemma.4}, a compact representation of $L^*$ can constructed in polynomial time. It is also clear that the function $\hat{d}_\mathrm{sum}$ can be computed efficiently. Then, by Theorem \ref{theorem:sfmDistributiveLattices} and Lemmas \ref{lemma.1}-\ref{lemma.4}, the proof of Theorem \ref{theorem.1} is complete. In the next section, we apply this framework to concrete combinatorial problems, verifying that their feasible solution sets satisfy properties \ref{property.1}-\ref{property.3} of the theorem. 

\section{Applications of the Framework} \label{sec.4}

We present examples of combinatorial problems whose feasible solution sets meet each of the conditions outlined in Theorem \ref{theorem.1}, allowing for the generation of maximally diverse solutions within our framework. Specifically, we discuss minimum $s$-$t$ cuts (Section \ref{section.mincuts}) and stable matchings (Section \ref{section.stableMatchings}). 

\subsection{\texorpdfstring{Minimum $s$-$t$ cuts}{Minimum s-t Cuts}} \label{section.mincuts}
The \textsc{Minimum $s$-$t$ Cut} problem is a classic combinatorial optimization problem. Given a directed graph $G = (V, E)$ and two special vertices $s, t \in V$, the problem asks for a subset $S \subseteq E$ of minimum cardinality that separates vertices $s$ and $t$, meaning that removing these edges from $G$ ensures there is no path from $s$ to $t$. Such a set is called a \textit{minimum $s$-$t$ cut} or \textit{$s$-$t$ mincut}. Here, we consider the problem of finding diverse minimum $s$-$t$ cuts, formally defined below. 

\begin{extthm}[\textsc{Max-Sum $k$-Diverse Minimum s-t Cuts}] \label{problem:2}
Given are a directed graph $G = (V, E)$ and vertices $s,t \in V$. Let $\Gamma \subseteq 2^E$ be the set of minimum $s$-$t$ cuts in $G$, and let $\Gamma_k$ be the set of $k$-multisets over $\Gamma$. Find $C \in \Gamma_k$ such that $d_\mathrm{sum}(C) = \max_{S \in \Gamma_k} d_\mathrm{sum}(S)$.
\end{extthm}

Using our framework, we reproduce the findings of De Berg et al. \cite{de2023finding} for \textsc{Max-Sum $k$-Diverse Minimum s-t Cuts}; that is, we show that the problem can be solved in polynomial time. For this, it suffices to show that the set $\Gamma$ of minimum $s$-$t$ cuts satisfies properties \ref{property.1}-\ref{property.3} of Theorem \ref{theorem.1}. We prove these statements in order. 

\begin{lemma}[Property 1] \label{lemma.MinCuts.1}
    There is a chain decomposition of the edge set $E$ into $r$ disjoint chains, such that each minimum $s$-$t$ cut $X \subseteq E$ contains exactly one element from each chain.  
\end{lemma}
\begin{proof}
    We construct the $r$ chains as follows. Let $\mathcal{P}$ be an (arbitrary) set of edge-disjoint $s$-$t$ paths in $G$ with maximum cardinality $r$. Define $E(p_i)$ as the set of edges traversed by the path $p_i \in \mathcal{P}$. For each path $p_i \in \mathcal{P}$, consider the order relation $\preceq_i$ defined as follows: for any $x, y \in E(p_i)$, we say $x \preceq_i y$ if and only if path $p_i$ meets edge $x$ before edge $y$, or if $x$ and $y$ are the same edge. Since every pair of edges within a path $p_i$ is comparable under this relation, each poset $(E(p_i), \preceq_i)$, for $i \in [r]$, forms a chain. Moreover, these chains are disjoint by the definition of the set $\mathcal{P}$. 

    By Menger's theorem, the size of a minimum $s$-$t$ cut in $G$ equals the maximum number of edge-disjoint $s$-$t$ paths, which is $r$. Consequently, any minimum $s$-$t$ cut $X \subseteq E$ must include exactly one edge from each chain $(E(p_i), \preceq_i)$, $i \in [r]$. Otherwise, if $X$ contained fewer than $r$ edges, it would not be a valid $s$-$t$ cut, and if it contained more, it would not be of minimum size. Hence, $\Gamma \subseteq E(p_1) \times \dots \times E(p_r)$. 

    Consider now the edges in $E' = E \setminus \bigcup_{1 \leq i \leq r} E(p_i)$. We call these edges \textit{residual edges}. Observe that these edges can never be part of a minimum $s$-$t$ cut. This follows because such a cut must contain exactly one edge from each chain in $\mathcal{P}$, and cutting any additional edge from $E'$ would only increase the cut size, violating minimality. Hence, we simply distribute the residual edges arbitrarily over the $r$ chains. This does not change the fact that the chains are disjoint, or that the set of minimum $s$-$t$ cuts is a subset of the cartesian product of the augmented chains. This completes the proof. 
\end{proof}

\begin{figure}
    \centering
    \begin{tikzpicture}[scale=0.9, transform shape]
    \tikzset{edge/.style = {->,> = latex}}
            \node[draw=black,circle,inner sep=2.5pt,minimum size=1pt] (v1) at (-1,0) {};
        \node[draw=black,circle,inner sep=2.5pt,minimum size=1pt] (v2) at (0,0) {};
        \node (dots1) at (1, 0) {$\dots$};
        \node[draw=black,circle,inner sep=2.5pt,minimum size=1pt] (w1) at (2,0) {};
        \node[draw=black,circle,inner sep=2.5pt,minimum size=1pt] (w2) at (3,0) {};
        \node[label=right:{$t$},draw=black,circle,fill=black,inner sep=2.5pt,minimum size=1pt] (dots2) at (4, 0) {};
        \draw[edge] (-2,0) node[left,label=left:{$s$},draw=black,circle,fill=black,inner sep=2.5pt,minimum size=1pt] {} -- (v1);

        \draw[edge,blue] (v1) -- node[above] {$x$} (v2);
    \draw[edge] (v2) -- (dots1);
    \draw[edge] (dots1) -- (w1);
    \draw[edge, blue] (w1) -- node[above] {$y$} (w2);
    \draw[edge] (w2) -- (dots2);
        \end{tikzpicture}
    \caption{Illustration of the order relation $\preceq_i$ over the edges of an $s$-$t$ path $p_i \in \mathcal{P}$.}
    \label{fig:posetPath}
\end{figure}
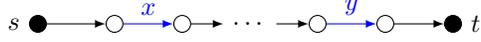 

By Lemma \ref{lemma.MinCuts.1}, the set $\Gamma \subseteq E(p_1) \times \dots \times E(p_r)$ of minimum $s$-$t$ cuts with componentwise order---defined by: $(x_1, \ldots, x_r) \preceq (y_1, \ldots, y_r)$ for $(x_1, \ldots, x_r), (y_1, \ldots, y_r) \in \Gamma$ iff $x_i \preceq_i y_i$ for all $i \in [r]$---forms a poset $L = (\Gamma, \preceq)$. It is well known that this poset defines a distributive lattice \cite{escalante1972schnittverbande, meyer1982lattices, halin1993lattices}. Specifically, proving that $\Gamma$ is closed under the joins and meets induced by $\preceq$ suffices to establish this property (see e.g., \cite[Claim A.1]{de2023finding}). Thus, property \ref{property.2} of Theorem \ref{theorem.1} follows directly. 

\begin{lemma}[Property 2] \label{lemma.MinCuts.2}
    The set $\Gamma$ of minimum $s$-$t$ cuts with componentwise order $\preceq$ defines a distributive lattice $L$. 
\end{lemma}

Next, we establish that a compact representation of the lattice of minimum $s$-$t$ cuts can be constructed in polynomial time. As with Lemma \ref{lemma.MinCuts.2}, this result is well known from the work of Picard and Queyranne \cite{picard1980structure}, who provided an algorithm to construct this representation via a so-called residual graph. 

\begin{lemma}[Property 3] \label{lemma.MinCuts.3}
    Let $L$ be the distributive lattice of $s$-$t$ mincuts in a graph $G$, there is a compact representation $G(L)$ of $L$ with the following properties: 
    \begin{enumerate}
        \item The vertex set is $J(L) \cup 0_L$,
        \item $|G(L)| \leq |V(G)|$,
        \item Given $G$ as input, $G(L)$ can be constructed in $O(|V(G)|^2)$ time. 
    \end{enumerate}
\end{lemma}

Then, by Theorem \ref{theorem.1} and Lemmas \ref{lemma.MinCuts.1}-\ref{lemma.MinCuts.3}, we obtain a polynomial time algorithm for \textsc{Max-Sum $k$-Diverse Minimum s-t Cuts} via submodular function minimization. 

\begin{theorem} \label{theorem.2}
    \textsc{Max-Sum $k$-Diverse Minimum s-t Cuts} can be solved in polynomial time. 
\end{theorem}

\begin{remark}
Similar results to those presented in Lemmas \ref{lemma.MinCuts.1}-\ref{lemma.MinCuts.3} can be established for \textit{minimum $s$-$t$ vertex cuts}. Since a vertex-connectivity version of Menger's theorem also exists, the arguments in Lemma \ref{lemma.MinCuts.1} remain valid when replacing $E$ with $V$. Moreover, the poset of minimum $s$-$t$ vertex cuts, ordered componentwise, forms a distributive lattice, which can be demonstrated analogously to Lemma \ref{lemma.MinCuts.2}. Finally, a compact representation of this lattice can be computed in polynomial time, as shown by Bonsma \cite[Sec. 6]{bonsma2010most}, or via the constructive version of Birkhoff’s theorem for computing a \textit{slice}, as described by Garg \cite{mittal2001computation} (see also \cite[Ch. 10]{garg2015introduction}).
\end{remark}

\subsection{Stable Matchings} \label{section.stableMatchings}
Finding a matching is one of the most fundamental combinatorial problems in graphs. Given a graph $G = (V, E)$, a \textit{matching} is any subset $M \subseteq E$ of edges such that no two edges in $M$ have a common endpoint. A matching $M$ is called a \textit{perfect matching} if every vertex in $G$ is incident to an edge in $M$. 

In the \textsc{Stable Matching} problem, we are given a complete bipartite graph $K_{n,n} = (A \cup B, E)$ along with a linear ordering $\preceq_a$ over $B$ for each vertex $a \in A$, and similarly a linear ordering $\preceq_b$ over $A$ for each vertex $b \in B$. For a vertex $a \in A$ (resp. $b \in B$), the poset $L_a = (B, \preceq_v)$ (resp. $L_b = (A, \preceq_v)$) is referred to as its \textit{preference list}. The task is to find a perfect matching $M$ in $K_{n, n}$ such that no two vertices $a \in A$ and $b \in B$ prefer each other over their matched partners. Such a matching is called a \textit{stable matching}. 

The \textsc{Stable Matching} problem models a wide range of real-world problems where two disjoint sets of entities are to be matched based on strict preferences (see e.g., \cite{Roth_Sotomayor_1990} for an overview). In this section, we consider the problem of finding diverse stable matchings. 

\begin{extthm}[\textsc{Max-Sum $k$-Diverse Stable Matching}] \label{problem:3}
Given are a complete bipartite graph $K_{n,n} = (A \cup B, E)$, along with preference lists $L_a$ and $L_b$ for each $a \in A$ and $b \in B$. Let $\Gamma \subseteq 2^E$ be the set of stable matchings in $G$, and let $\Gamma_k$ denote the set of $k$-multisets over $\Gamma$. Find $C \in \Gamma_k$ such that $d_\mathrm{sum}(C) = \max_{S \in \Gamma_k} d_\mathrm{sum}(S)$.
\end{extthm}

We now show that \textsc{Max-Sum $k$-Diverse Stable Matching} can be solved in polynomial time by proving that the set $\Gamma$ of stable matchings satisfies properties \ref{property.1}-\ref{property.3} of Theorem \ref{theorem.1}. 

\begin{lemma}[Property 1] \label{lemma.StableMatching.1}
    There is a chain decomposition of the edge set $E$ into $r$ disjoint chains, such that each stable matching $X \subseteq E$ contains exactly one element from each chain.  
\end{lemma}
\begin{proof}
    Let $r = n$. Note that the posets $L_a$ and $L_b$ are chains. We claim that the chains $L_a$, $a \in A$ define a disjoint chain decomposition of the ground set $E$.\footnote{Note that we may also choose the chains $L_b$, $b \in B$ and get similar results.} First, we argue for disjointness. Let $E(a) = \{(a, b) \; | \; b \in L_a\}$ denote the set of edges defined by the preference list $L_a$ of an arbitrary vertex $a \in A$. Since $K_{n, n}$ is bipartite, there are no edges between the vertices of $A$. This implies that $E(a_1) \cap E(a_2) = \emptyset$ for all distinct $a_1, a_2 \in A$. Moreover, $E = \bigcup_{a \in A} E(a)$. Hence, the chains $L_a$, $a \in A$ define a disjoint chain decomposition of $E$. 
    
    Now, we argue that a stable matching must contain exactly one element from each chain $L_a$. This follows immediately from the definition of perfect matching, which requires every vertex in $A$ to be matched to exactly one vertex in $B$. Consequently, each stable matching selects precisely one edge from $E(a)$ for each $a \in A$. This completes the proof of the lemma. 
\end{proof}

To verify property \ref{property.2}, we use the following well-established result from the stable matchings literature \cite{knuth1997stable, blair1988lattice}. 

\begin{claim}[{\cite[Thm. 7 \& Cor. 1]{knuth1997stable}}] \label{claim.StableMatching.1}
    Given any two stable matchings $X = ((a_1, b_1), \ldots, (a_n, b_n))$ and $Y = ((a_1, b_1'), \ldots, (a_n, b_n'))$, then  
    \begin{align*}
        X \vee Y & = ((a_1, \max_{\preceq_{a_1}}(b_1, b_1')), \ldots, (a_n, \max_{\preceq_{a_n}}(b_n, b_n')))) \quad \text{and} \\
        X \wedge Y & = ((a_1, \min_{\preceq_{a_1}}(b_1, b_1')), \ldots, (a_n, \min_{\preceq_{a_n}}(b_n, b_n'))))
    \end{align*}
    are also stable matchings.
\end{claim}

By standard results in lattice theory (see e.g., \cite{gratzer2009lattice}), the cartesian product $E_\mathrm{prod} = E(a_1) \times \dots \times E(a_n)$, with componentwise order $\preceq$, forms a distributive lattice $(E_\mathrm{prod}, \preceq)$. Then, by Lemma \ref{lemma.StableMatching.1} and Claim \ref{claim.StableMatching.1}, the poset $L = (\Gamma, \preceq)$ is a sublattice of $E_\mathrm{prod}$, which implies that $L$ is also distributive.

\begin{lemma}[Property 2] \label{lemma.StableMatching.2}
    The set $\Gamma$ of stable matchings with componentwise order $\preceq$ defines a distributive lattice $L$. 
\end{lemma}

It only remains to verify that property \ref{property.3} of Theorem \ref{theorem.1} is satisfied by the set $\Gamma$ of stable matchings. As with minimum $s$-$t$ cuts, this property follows directly, since the so-called \textit{poset of rotations} introduced by Gusfield \cite{gusfield1989stable} provides the required structure (see also, e.g., \cite[Sec 2.3]{gangam2022structural}).  

\begin{lemma}[Property 3 {\cite[Lemma 3.3.2]{gusfield1989stable}}] \label{lemma.StableMatching.3}
    A compact representation of the lattice $L$ of stable matchings can be constructed in $O(|V|^2)$ time. 
\end{lemma}

Then, by Theorem \ref{theorem.1} and Lemmas \ref{lemma.StableMatching.1}-\ref{lemma.StableMatching.3}, the following theorem holds. 

\begin{theorem} \label{theorem.3}
    \textsc{Max-Sum $k$-Diverse Stable Matching} can be solved in polynomial time. 
\end{theorem}


\section{Other Diversity Measures} \label{sec.5}
The proof of Theorem \ref{theorem.1} relies on four lemmas, with the diversity measure playing a role in only two of them: the Cost-Equivalence (Lemma \ref{lemma.2}) and Submodularity (Lemma \ref{lemma.3}) lemmas. For simplicity, we have presented our main result in terms of the $d_\mathrm{sum}$ diversity measure. However, the framework is not limited to this specific measure. Just as it applies to problems whose solution sets satisfy the properties of Theorem \ref{theorem.1}, it also extends to other diversity measures, provided they satisfy both the Cost-Equivalence and Submodularity lemmas.

Here, we mention two such diversity measures: the \textit{coverage} diversity $d_\mathrm{cov}$, and the $L_1$- or \textit{absolute-difference} diversity $d_{abs}$. Let $E$ be a finite set with $n$ elements, and let $\Gamma \subseteq 2^E$ be a set of feasible solutions. Given a $k$-tuple of feasible solutions $(X_1, \ldots, X_k) \in \Gamma^k$, these measures are defined as follows:
\begin{align}
    d_\mathrm{cov}(X_1, X_2, \ldots, X_k) & = \bigcup_{1 \leq i \leq k} |X_i|, \quad \text{and} \label{eq:coverage} \\
    d_\mathrm{abs}(X_1, X_2, \ldots, X_k) & = \sum_{1 \leq i < j \leq k} f(X_i, Y_j), \label{eq:L1}
\end{align}
where $f(X, Y) = \sum_i^r \left\|x_i - y_i\right\|$ for any two $X = (x_1, \ldots, x_r), Y = (y_1, \ldots, y_r) \in \Gamma$.

The coverage diversity measures the number of distinct elements appearing across solutions, while the absolute-difference diversity quantifies diversity by summing coordinate-wise differences between solutions. We remark that the absolute-difference measure is defined only for solutions that can be represented as an $r$-tuple, as the function $f$ requires component-wise comparisons between elements in the ground set $E$. Moreover, it requires a notion of difference between elements of $E$ (e.g., the set of $r$-dimensional integer vectors in $[-M, M]^r$, with $M \in \mathbb{N}$). 

In the remainder of this section, we demonstrate that $d_\mathrm{cov}$ (Section \ref{sec.cov}) and $d_\mathrm{abs}$ (Section \ref{sec.abs}) satisfy the Cost-Equivalence and Submodularity lemmas of Section \ref{sec.3}, and thus can be used as optimization objectives to generate diverse solutions efficiently.

\subsection{Coverage Diversity} \label{sec.cov}

We consider the problem of finding maximally diverse solutions with respect to the coverage diversity measure, defined formally as follows. 

\begin{extthm}[\textsc{Max-Cov $k$-Diverse Solutions}] \label{problem:6} 
Given a finite set $E$ of size $n$ and a membership oracle for $\Gamma \subseteq 2^E$, find a $k$-element multiset $C = (X_1, X_2, \ldots, X_k)$ with $X_1, X_2, \ldots, X_k \in \Gamma$, such that $d_\mathrm{cov}(C)$ is maximum. 
\end{extthm}

We prove the following result by means of the reduction to submodular function minimization established in Section \ref{sec.3}. 

\begin{theorem} \label{theorem.5}
    \textsc{Max-Cov $k$-Diverse Solutions} can be solved in polynomial time if the set of feasible solutions $\Gamma$ satisfies the three properties of Theorem \ref{theorem.1}. 
\end{theorem} 

It suffices to show that the Cost-Equivalence and Submodularity lemmas hold for the $d_\mathrm{cov}$ measure. We prove these statements in order, closely following the work of De Berg et al. \cite{de2023finding}. 

\paragraph{Proof of Cost-Equivalence.}
Let $C \in \Gamma^k$ be a $k$-tuple of solutions. Recall that the multiplicity $\mu_e(C)$ of an element $e \in E$, with respect to $C$, is the number of feasible solutions in $C$ that contain $e$. By property \ref{property.1} of Theorem \ref{theorem.1}, we can rewrite the coverage diversity directly in terms of the multiplicity as 
\begin{equation}
    d_{\text{cov}}(C) = k \cdot r - \left(\sum_{e \in E_{\mathrm{shr}}(C)} \left( \mu_e(C) - 1 \right) \right), \label{eq:dcovMultiplicity}
\end{equation}
where $E_\mathrm{shr}(C) \subseteq E$ denotes the subset of elements whose multiplicity satisfies $\mu_e(C) \geq 2$. Consider an element $e \in E$. If $e \in E_\mathrm{shr}$, we say that $e$ is a \textit{shared} element; otherwise, it is an \textit{unshared} element.  

By Claim \ref{claim.0}, any $k$-tuple in $\Gamma^k$ can be reordered into a left-right ordered form while preserving element multiplicities. Moreover, Equation \eqref{eq:dcovMultiplicity} establishes that the diversity $d_\mathrm{cov}$ of a collection of solutions depends solely on element multiplicities since the terms outside the summation are constant. Then, any $k$-tuple can be reordered into a left-right ordered form with the same coverage diversity value. This implies that $d_\mathrm{cov}$ satisfies the Cost-equivalence lemma. 

\begin{lemma}[Cost-Equivalence Lemma] \label{lemma.Cov.1}
    For every $k$-tuple $C \in \Gamma^k$ there exists $\hat{C} \in \Gamma_{\mathrm{lr}}^k$ such that $d_{\mathrm{cov}}(\hat{C}) = d_{\mathrm{cov}}(C)$.
\end{lemma}

\paragraph{Proof of Submodularity.}
Next, observe that maximizing $d_\mathrm{cov}$ is equivalent to minimizing $\hat{d}_\mathrm{cov}$, defined as
\begin{equation}
    \hat{d}_\mathrm{cov}(C) = \sum_{e \in E_{\mathrm{shr}}(C)} \left( \mu_e(C) - 1 \right). \label{eq:dcovMultiplicityMin}
\end{equation}

Given the Claims \ref{claim.1}-\ref{claim.3} of Section \ref{sec.3}, we now show that $\hat{d}_\mathrm{cov}(C)$ is submodular on the lattice $L^* = (\Gamma^k_\mathrm{lr}, \preceq^k)$ of left-right ordered $k$-tuples of feasible solutions. 

First, consider the function $F_e(C) : \Gamma_{\mathrm{lr}}^k \rightarrow \mathbb{N}$ defined by $F_e(C) = \mu_e(C)-1$. It is an immediate corollary of Claim \ref{claim.1} that $F_e(C)$ is modular in $L^*$. Then, the sum $\sum_e F_e(C)$ taken over all elements $e \in E$ is still a modular function. Notice that only shared elements in $C$ contribute positively to the sum, while the contribution of unshared elements can be neutral or negative. We split this sum into two parts: the sum over shared elements $e \in E_\mathrm{shr}(C)$, and the sum over unshared elements $e \in E \setminus E_\mathrm{shr}(C)$. The latter sum can be further simplified to $|E(C)| - |E|$ by observing that only unshared elements make a (negative) contribution. Therefore, we can write
\begin{equation} \label{eq:sumCoverage}
     \sum\nolimits_{e \in E} F_e(C) = \left( \sum\nolimits_{e \in E_\mathrm{shr}(C)} (\mu_e(C) - 1)\right) + |E(C)| - |E|.
\end{equation}
We know $\sum_e F_e(C)$ to be a modular function on $L^*$, hence for any two $C_1, C_2 \in L^*$ we have
\begin{equation*}
    \sum_{e \in E} F_e(C_1 \vee C_2) + \sum_{e \in E} F_e(C_1 \wedge C_2) = \sum_{e \in E} F_e(C_1) + \sum_{e \in E} F_e(C_2),
\end{equation*}
which, by equation \eqref{eq:sumCoverage}, is equivalent to
\begin{align} \label{eq:coverageEquality}
    & \left( \sum_{e \in E_\mathrm{shr}(C_1 \vee C_2)} (\mu_e(C_1 \vee C_2) - 1)  + \sum_{e \in E_\mathrm{shr}(C_1 \wedge C_2)} (\mu_e(C_1 \wedge C_2) - 1) \right) + |E(C_1 \vee C_2)| + |E(C_1 \wedge C_2)| \notag \\ 
    & = \left( \sum_{e \in E_\mathrm{shr}(C_1)} (\mu_e(C_1) - 1) + \sum_{e \in E_\mathrm{shr}(C_2)} (\mu_e(C_2) - 1) \right) + |E(C_1)| + |E(C_2)|. 
\end{align}
Now, from Claims \ref{claim.1} and \ref{claim.3}, we make the following proposition (see also \cite[Sec. A.5]{de2023finding}). 

\begin{proposition} \label{claim:edgeSetSizesInequality}
For any two $C_1, C_2 \in L^*$ we have $|E(C_1 \vee C_2)| + |E(C_1 \wedge C_2)| \geq |E(C_1)| + |E(C_2)|$.
\end{proposition}
\begin{proof}
    For simplicity, we denote sets $E(C_1 \vee C_2)$, $E(C_1 \wedge C_2)$, $E(C_1)$, and $E(C_2)$ by $A$, $B$, $C$, and $D$, respectively. Consider an element $e \in E$. We begin with two simple facts that can be derived from Claims \ref{claim.1} and \ref{claim.3}. First, if $e \in A \triangle B$ then $e \in C \triangle D$. Second, if $e \in C \cap D$, then $e \in A \cap B$. We remark that the reverse direction in each of these facts is not always true. Let $X = (C \triangle D) \cap (A \cap B)$ be the set of elements that appear in $C \setminus D$ or $D \setminus C$ and also appear in $A \cap B$. We may then write: 
    \begin{align*}
    |C| + |D| & = |C \setminus D| + |D \setminus C| + 2 |C \cap D| \notag \\
    & = |A \setminus B| + |B \setminus A| + |X| + 2 |A \cap B| - 2 |X| \notag \\
    & = |A| + |B| - |X| \notag \\ 
    & \leq |A| + |B|
\end{align*}
where the second inequality follows from observing that $|C \setminus D| + |D \setminus C| = |A \setminus B| + |B \setminus A| + |X|$ and $|A \cap B| = |C \cap D| + |X|$. 
\end{proof}

Given Claim \ref{claim:edgeSetSizesInequality}, it is clear that to satisfy equality in equation \eqref{eq:coverageEquality} it must be that: 
\begin{align*} 
    \sum_{e \in E_\mathrm{shr}(C_1 \vee C_2)} (\mu_e(C_1 \vee C_2) - 1) & + \sum_{e \in E_\mathrm{shr}(C_1 \wedge C_2)} (\mu_e(C_1 \wedge C_2) - 1) \\ 
    & \leq \sum_{e \in E_\mathrm{shr}(C_1)} (\mu_e(C_1) - 1) + \sum_{e \in E_\mathrm{shr}(C_2)} (\mu_e(C_2) - 1),
\end{align*}
from which the submodularity of $\hat{d}_\mathrm{cov}$ immediately follows.

\begin{lemma}[Submodularity Lemma] \label{lemma.Cov.2}
The function $\hat{d}_\mathrm{cov}: \Gamma_{\mathrm{lr}}^k \rightarrow \mathbb{N}$ is submodular in $L^*$.
\end{lemma}

Replacing the Cost-Equivalence and Submodularity lemmas of Section \ref{sec.3} with Lemmas \ref{lemma.Cov.1} and \ref{lemma.Cov.2} above, the proof of Theorem \ref{theorem.5} is complete. 

\subsection{Absolute-Difference Diversity} \label{sec.abs}
In this section, we consider the problem of finding maximally diverse solutions with respect to the absolute-difference diversity measure. In the following problem definition, we already assume that the property \ref{property.1} of Theorem \ref{theorem.1} holds, since $d_\mathrm{abs}$ requires componentwise comparisons. Furthermore, we require a transitive binary difference operation between componentwise elements.\footnote{This is related to the notion of \textit{difference poset} or \textit{D-poset} \cite{kopka1994d}.} 

\begin{extthm}[\textsc{Max-Abs $k$-Diverse Solutions}] \label{problem:7} 
Given a poset $E$ of $n$ elements formed by the disjoint union of $r$ chains $(E_i, \preceq_i)$, $i \in [r]$, and equipped with a binary difference operation ($-$), along with a membership oracle for $\Gamma \subseteq E_1 \times \dots \times E_r$, find a $k$-element multiset $C = (X_1, X_2, \ldots, X_k)$ with $X_1, X_2, \ldots, X_k \in \Gamma$, such that $d_\mathrm{abs}(C)$ is maximum. 
\end{extthm}

We prove the following result by means of the reduction to submodular function minimization established in Section \ref{sec.3}. 

\begin{theorem} \label{theorem.6}
    \textsc{Max-Abs $k$-Diverse Solutions} can be solved in polynomial time if the set of feasible solutions $\Gamma$ satisfies the three properties of Theorem \ref{theorem.1}. 
\end{theorem} 

Like before, it suffices to show that the Cost-Equivalence and Submodularity lemmas hold for the absolute difference measure. 

\paragraph{Proof of Cost-Equivalence.} The following lemma is an immediate consequence of Claim \ref{claim.0}, which states that any $k$-tuple in $\Gamma^k$ can be reordered into a left-right ordered form while preserving element multiplicities. 

\begin{lemma}[Cost-Equivalence Lemma] \label{lemma.Abs.1}
    Let $C \in \Gamma^k$ such that $d_{\text{abs}}(C) = \max_{S \in \Gamma^k} d_{\mathrm{abs}}(S)$. Then there exists $\hat{C} \in \Gamma_{\mathrm{lr}}^k$ such that $d_{\mathrm{abs}}(\hat{C}) = d_{\mathrm{abs}}(C)$. 
\end{lemma}
\begin{proof}
Let $C \in \Gamma^k$ be an arbitrary $k$-tuple of solutions, and let $\hat{C} \in \Gamma^k_\mathrm{lr}$ be its reordering into left-right order by the algorithm of Claim \ref{claim.0}. 
For a feasible solution $X \in \Gamma$, let $X(\ell)$ denote its $\ell$-th component. 

Consider the $k$-tuples 
$C(\ell) = (X_1(\ell), \ldots, X_k(\ell))$ and $\hat{C}(\ell) = (\hat{X}_1(\ell), \ldots, \hat{X}_k(\ell))$, where $X_i \in C$, $\hat{X}_i \in \hat{C}$ for all $i \in [k]$. These represent the $\ell$-th component of each solution in $C$ and $\hat{C}$, respectively. Now, define the function $f_{\ell} : E_\ell^k \rightarrow \mathbb{R}$ as:
\[
f_{\ell}(x_1, \ldots, x_k) = \sum_{1 \leq i < j \leq k} \left\|x_i - x_j\right\|.
\]
By Claim \ref{claim.0}, the multiplicity of each element in $C(\ell)$ is preserved in $\hat{C}(\ell)$, implying that $f_{\ell}(C(\ell)) = f_{\ell}(\hat{C}(\ell))$. Since the absolute-difference diversity measure decomposes as:
\[
d_\mathrm{abs}(X_1, \ldots, X_k) = \sum_{\ell=1}^{r} f_{\ell}(X_1(\ell), \ldots, X_k(\ell)),
\]
it follows that $d_{\mathrm{abs}}(\hat{C}) = d_{\mathrm{abs}}(C)$. In particular, this holds for tuples that achieve maximum diversity.
\end{proof}

\paragraph{Proof of Submodularity.} In this case, the Submodularity Lemma actually becomes a Modularity Lemma. First, consider the function $f_{\ell}' : \Gamma^2_\mathrm{lr} \rightarrow \mathbb{R}$ defined by $f_{\ell}'(X_1, X_2) = \left\| X_1(\ell) - X_2(\ell) \right\|$, where $X_1 \preceq X_2$ and $\ell \in [r]$. We can rewrite the absolute difference diversity measure as: 
\[
d_\mathrm{abs}(X_1, \ldots, X_k) = \sum_\ell^r \sum_{1 \leq i < j \leq k} f_{\ell}'(X_i, X_j).
\]

If we establish that $f_{\ell}'(\cdot)$ is modular in $L$ then, because the sum of modular functions is modular, $d_\mathrm{abs}$ is also modular. This is done in the following lemma. 

\begin{lemma}[Modularity Lemma] \label{lemma.Abs.2}
The function $\hat{d}_\mathrm{abs}: \Gamma_{\mathrm{lr}}^k \rightarrow \mathbb{N}$ is modular in $L^*$.
\end{lemma}
\begin{proof}We prove that, for any two $C_1, C_2 \in \Gamma^2_\mathrm{lr}$, the function $f_{\ell}'(\cdot)$ is modular in the lattice $(\Gamma^2_\mathrm{lr}, \preceq^2)$, where $\preceq$ is the componentwise order of the poset $L$ of feasible solutions. Since the sum of modular functions is modular, and $d_\mathrm{abs}$ can be written as a sum of functions $f_{\ell}'(\cdot)$ over all $\ell \in [r]$, the modularity of $d_\mathrm{abs}$ follows.

Let $C_1 = (X_1, X_2)$ and $C_2 = (Y_1, Y_2)$. By definition, the join ($\vee$) and meet ($\wedge$) of $C_1$ and $C_2$ are given by componentwise maximum and minimum. Then, for each $\ell \in [r]$, 
\begin{align*}
f_\ell'(C_1 \wedge C_2) & = \left\| \min(X_1(\ell), Y_1(\ell)) - \min(X_2(\ell), Y_2(\ell)) \right\| \quad \text{and} \\
f_\ell'(C_1 \vee C_2) & = \left\| \max(X_1(\ell), Y_1(\ell)) - \max(X_2(\ell), Y_2(\ell)) \right\|.
\end{align*}

Consider an arbitrary $\ell \in [r]$. Because $C_1$ and $C_2$ are each in left-right order, we have: $X_1(\ell) \preceq X_2(\ell)$ and $Y_1(\ell) \preceq Y_2(\ell)$. Consider then the intervals $I_X = [X_1(\ell), X_2(\ell)]$ and $I_Y = [Y_1(\ell), Y_2(\ell)]$. Without loss of generality, assume that $X_2(\ell) \preceq Y_2(\ell)$. There are three possibilities for the interaction of $I_X$ and $I_Y$: (i) the intervals are disjoint (i.e., $I_X \cap I_Y = \emptyset$), they overlap (i.e., $I_X \cap I_Y \neq \emptyset$), or (iii) one is contained in the other (i.e., $I_X \subset I_Y$). We now compare the sums $f_\ell'(C_1 \wedge C_2) + f_\ell'(C_1 \vee C_2)$ and $f_\ell'(C_1) + f_\ell'(C_2)$ in each of these cases. 

In cases (i) and (ii), we have that $X_1(\ell) \preceq Y_1(\ell)$ and $X_2(\ell) \preceq Y_2(\ell)$. Hence, 
\[
f_\ell'(C_1 \wedge C_2) + f_\ell'(C_1 \vee C_2) = \left\| X_1(\ell) - X_2(\ell) \right\| + \left\| Y_1(\ell) - Y_2(\ell) \right\| = f_\ell'(C_1) + f_\ell'(C_2),
\]
and thus, modularity is satisfied. 

In case (iii), we have $Y_1(\ell) \preceq X_1(\ell)$ and $X_2(\ell) \preceq Y_2(\ell)$. Then:
\begin{align*}
    f_\ell'(C_1 \wedge C_2) + f_\ell'(C_1 \vee C_2) & = \left\| Y_1(\ell) - X_2(\ell) \right\| + \left\| X_1(\ell) - Y_2(\ell) \right\| \\ 
    & = (X_2(\ell) - Y_1(\ell)) + (Y_2(\ell) - X_1(\ell)) \\
    & = (X_2(\ell) - X_1(\ell)) + (Y_2(\ell) - Y_1(\ell)) \\
    & = f_\ell'(C_1) + f_\ell'(C_2),
\end{align*}
which again satisfies modularity. 

Therefore, the function $f_{\ell}'(\cdot)$ is modular in $(\Gamma^2_\mathrm{lr}, \preceq^2)$, and the lemma is proved. 
\end{proof}

By replacing the Cost-Equivalence and Submodularity lemmas of Section \ref{sec.3} with Lemmas \ref{lemma.Abs.1} and \ref{lemma.Abs.2} above, the proof of Theorem \ref{theorem.6} is complete. 

\section{A Simple Framework for Disjoint Solutions} \label{sec.6}
In this section, we consider the problem of finding a largest set of pairwise disjoint solutions. More precisely, we present an algorithm for solving \textsc{Max-Disjoint Solutions}, formally defined below. 

\begin{extthm}[\textsc{Max-Disjoint Solutions}] \label{problem:8}
Given a finite set $E$ of size $n$, an implicitly defined family $\Gamma$ of subsets of $E$, referred to as feasible solutions, and a membership oracle $\mathcal{O}_{\Gamma}$ for $\Gamma$, find a set $C \subseteq \Gamma$ such that $X \cap Y = \emptyset$ for all $X, Y \in C$, and $|C|$ is as large as possible. 
\end{extthm}

We assume that the set $\Gamma$ of feasible solutions satisfies the structural properties \ref{property.1} and \ref{property.2} of Theorem \ref{theorem.1}. That is, there is a poset $P = (E, \leq)$ that is the disjoint union of $r$ chains $(E_i, \preceq_i)$, $i \in [r]$, and the set $\Gamma \subseteq E_1 \times \dots \times E_r$, with componentwise order $\preceq$, forms a distributive lattice. 

The idea behind the algorithm is simple: start by finding the bottom element of the lattice of feasible solutions, remove it along with any other solutions that overlap with it, and then repeat this process on the remaining sublattice until no feasible solutions remain. Of course, we want to avoid working on the lattice directly, as it can be of exponential size. Instead, we assume that the algorithm has access to the following oracles, or subroutines:  
    \begin{itemize}
    \item Minimal/Maximal Solution Oracles ($\mathcal{O}_{\min}$ and $\mathcal{O}_{\max}$): On input $\langle P, \mathcal{O}_\Gamma \rangle$, the \textit{minimal solution oracle} ($\mathcal{O}_{\min}$) returns the bottom element of the distributive lattice $(\Gamma, \preceq)$, while the \textit{maximal solution oracle} ($\mathcal{O}_{\max}$) returns its top element; i.e., 
    \begin{equation*}
    \mathcal{O}_{\min}(P, \mathcal{O}_\Gamma) = \bigwedge_{X \in \Gamma} X, \quad \text{and} \quad \mathcal{O}_{\max}(P, \mathcal{O}_\Gamma) = \bigvee_{X \in \Gamma} X.
    \end{equation*}
    
    \item Disjoint Successors Oracle ($\mathcal{O}_\mathrm{ds}$): For a feasible solution $X \in \Gamma$, the subset $\Gamma(X) \subset \Gamma$ of \textit{disjoint successors of $X$} consists of all feasible solutions that are both disjoint from $X$ and successors of $X$ with respect to the order $\preceq$; that is, 
    $
    \Gamma(X) = \{ Y \mid Y \in \Gamma, X \cap Y = \emptyset, X \preceq Y \}.
    $
    Given an input $\langle X, P, \mathcal{O}_\Gamma \rangle$, this oracle returns the subposet of $P$ induced by the subset of elements of $E$ that appear in the disjoint succesors of $X$; i.e.,   
    \[
    \mathcal{O}_\mathrm{ds}(X, E, \mathcal{O}_\Gamma) = P \left[ \bigcup \Gamma(X) \right].
    \]
\end{itemize}

In this general framework, we achieve the following result. 

\begin{theorem} \label{theorem.7}
    \textsc{Max-Disjoint Solutions} can be solved in $O(n)$ oracle calls.  
\end{theorem}

It is important to note that these subroutines are problem specific, and must be designed and implemented based on the particular problem defined by $P$ and $\Gamma$. This framework has been implicitly applied in previous work to develop efficient algorithms for identifying collections of pairwise disjoint minimum $s$-$t$ cuts \cite{de2023finding} and stable matchings \cite{ganesh2021disjoint}, where the corresponding oracle implementations run in near-linear time. Here, we extend and generalize the ideas in these works to the broader class of problems that satisfy properties \ref{property.1} and \ref{property.2} of Theorem \ref{theorem.1}.

\paragraph{Preliminaries.} Before we formally describe the algorithm, we require some results and notation. Throughout, let $L$ denote the distributive lattice $(\Gamma, \preceq)$. We use $X_z$ and $X_o$ to denote the top and bottom elements of a lattice $L$, respectively, which are the two elements that satisfy $X_o \preceq X \preceq X_z$ for all $X \in \Gamma$. For a feasible solution $X \in \Gamma$, we use $X(\ell)$ to denote the element in the $\ell$-th component of $X$. Note that $X(\ell) \in E_\ell$. The following observation is a necessary condition for the existence of disjoint solutions in $\Gamma$. 

\begin{observation} \label{obs.2}
    Let $e \in E$ be an element of both $X_o$ and $X_z$. Then $e$ must be present in every feasible solution in $\Gamma$. 
\end{observation}
\begin{proof}
    Consider an arbitrary feasible solution $X \in \Gamma$. Without loss of generality, let $e \in E_\ell$. By definition of bottom element of $L$, we have $X_o \preceq X$, and thus $e \preceq_\ell X(\ell)$. On the other hand, by definition of top element of $L$, we have $X \preceq X_z$, which implies $X(\ell) \preceq_\ell e$. By antisymmetry of the partial order $\preceq_\ell$, it follows that $X(\ell) = e$. Hence, $e \in X$, proving the fact.
\end{proof}

We also make the following observation about the set of disjoint successors of a feasible solution. With a slight abuse of notation, we use $\preceq$ to denote to the componentwise ordering arising from $P$ and any induced suposet. 

\begin{observation} \label{obs.3}
    For any $X \in \Gamma$, the set $\Gamma(X)$ of disjoint successors of $X$ satisfies properties \ref{property.1} and \ref{property.2} of theorem \ref{theorem.1}. 
\end{observation}
\begin{proof}
    We start with the first property. Let $P(X) = \left[ \bigcup \Gamma(X) \right]$ denote the subposet induced by $\bigcup \Gamma(X)$. This subposet then consists of the disjoint chains of $P$ but restricted to the elements appearing in $\bigcup \Gamma(X)$. By definition of both $P(X)$ and $\Gamma(X)$, each solution in $\Gamma(X)$ must contain exactly one element from each chain in $P(X)$. Hence, property \ref{property.1} is satisfied. 

    As for the second property, it is clear that $\Gamma(X) \subset \Gamma$. Moreover, because the join ($\vee$) and meet ($\wedge$) operations in $L$ are defined as the componentwise maximum and minimum, respectively, $\Gamma(X)$ remains closed under these operations. 
    This means that the poset $(\Gamma(X), \preceq)$ is a sublattice of $L$ and thus, a distributive lattice. Hence, property \ref{property.2} is satisfied.
\end{proof}

With these results, we are ready to describe and analyze the algorithm. 

\paragraph{The algorithm.} Given an input $\langle P, \mathcal{O}_\Gamma \rangle$, the algorithm begins by determining the bottom element $X_o$ and the top element $X_z$ of lattice $L$ by querying the oracles $\oraclemin$ and $\oraclemax$ with the input $\langle P, \oraclegamma \rangle$. If these two solutions share an element, the algorithm stops, as Observation \ref{obs.2} ensures that no disjoint solutions exist. Otherwise, it proceeds by querying $\oracleds(X_o, P, \oraclegamma)$ to determine the subposet $P(X_o)$ induced by $\bigcup \Gamma(X_o)$. 

By Observation \ref{obs.3}, the set $\Gamma(X_o)$ satisfies properties \ref{property.1} and \ref{property.2} of Theorem \ref{theorem.1}, with the poset $P(X_o)$ serving as the corresponding chain decomposition. Let $L(X_o) = (\Gamma(X_o), \preceq)$ be the associated sublattice of disjoint successors of $X_o$. 
The algorithm proceeds by querying $\oraclemin$ with the input $\langle P(X_o), \oraclegamma \rangle$ to identify the bottom element $X_o'$ of $L(X_o)$. Once more, if $X_o'$ is disjoint from $X_z$, the algorithm queries $\oracleds(X_o', P(X_o), \oraclegamma)$ to determine the subposet $P(X_o')$ induced by the set $\bigcup \Gamma(X_o')$ of disjoint successors of $X_o'$. This process repeats as long as $\oraclemin$ continues to return solutions that are disjoint from $X_z$. Throughout the execution, the algorithm maintains a set $C$ that stores all solutions found that are disjoint from $X_z$ and returns this set upon termination. The algorithm is presented below as Algorithm \ref{algo.1}. 

\begin{algorithm}
\caption[Caption for LOF]{\textsc{Max-Disjoint Solutions}} \label{algo.1}
\vspace{.5em}
\textbf{Input:}  A poset $P$ and a membership oracle $\oraclegamma$ for $\Gamma$ satisfying properties \ref{property.1} and \ref{property.2} of Theorem \ref{theorem.1}. \\
\textbf{Output:}  A maximum cardinality set $C$ of disjoint feasible solutions from $\Gamma$. 
\vspace{.5em}
\begin{algorithmic}[1]
\newcommand\NoDo{\renewcommand\algorithmicdo{}}
\newcommand\ReDo{\renewcommand\algorithmicdo{\textbf{do}}}
\algrenewcommand\algorithmiccomment[1]{\hfill {\color{blue} \(\triangleright\) #1}}

\State $C \leftarrow \emptyset$
\State $X_z \leftarrow \oraclemax(P, \oraclegamma)$ \Comment{Top element of lattice $L$.}
\State $X \leftarrow \oraclemax(P, \oraclegamma)$  \Comment{Bottom element of lattice $L$.}
\State $P(X) \leftarrow \oracleds(P, X_o, \oraclegamma)$ \Comment{This defines a new instance.}
\While{$X \cap X_z = \emptyset$}
\State $C \leftarrow C \cup \{X\}$
\State $X \leftarrow \oraclemin(P(X), \oraclegamma)$    \Comment{New disjoint solution.}
\EndWhile
\State $C \leftarrow C \cup \{X\}$
\State \textbf{return} $C$
\end{algorithmic}
\vspace{.5em}
\end{algorithm}

\paragraph{Correctness.} The solutions in the set $C = \{X_1, X_2, \ldots, X_k\}$ returned by Algorithm \ref{algo.1} are clearly disjoint by construction, as the poset returned by the oracle $\oracleds$ at each step is induced by the set of disjoint successors of the solution identified in the precious step. Moreover, the set $C$ is, in fact, a left-right ordered tuple. This follows again by construction, as each newly identified solution is determined from the subset of elements that are chain-successors of elements included in previously identified solutions. Note that the notion of left-right order here is strict, meaning that $X_i \prec X_j$ for any $1 \leq i < j \leq k$. 

To analyze this further, let us go back for a moment to Section \ref{sec.3}. Note that the $d_\mathrm{sum}$ measure is maximum whenever its input consists of disjoint solutions. Consider then an arbitrary $k$-tuple of disjoint feasible solutions, for some $k > 0$. We know, by Claim \ref{claim.0}, that there exists a $k$-tuple of disjoint feasible solutions that is in left-right order. In particular, this is true for a disjoint-solutions tuple of maximum cardinality $k^*$. Then, as we did in Section \ref{sec.3}, we may restrict our arguments to the set of $k^*$-tuples that are in left-right order without loss of generality. 

To complete the correctness of Algorithm \ref{algo.1}, it remains to show that the tuple returned by the algorithm is of maximum cardinality $k^*$. 

\begin{lemma}
    Algorithm \ref{algo.1} outputs a longest tuple of disjoint feasible solutions. 
\end{lemma}
\begin{proof}
    Let $C_\mathrm{ALG} = (X_1, X_2, \ldots, X_k)$ be the $k$-tuple of disjoint feasible solutions returned by Algorithm \ref{algo.1}. For the sake of contradiction, suppose that $C' = (Y_1, Y_2, \ldots, Y_\ell)$ is a longest left-right ordered tuple of disjoint feasible solutions with $\ell > k$. 

    By definition of bottom element, we know that solution $X_1 = X_o$ is a predecessor of every other feasible solution in $\Gamma$. This implies that $Y_1 \cap X_1 \neq \emptyset$; otherwise, we could append $X_1$ to the start of $C'$ and obtain a longer tuple of left-right ordered disjoint solutions. Then, we have $X_1 \preceq Y_1 \prec Y_2$, and we may replace $Y_1$ in $C'$ with $X_1$ to generate a new $\ell$-tuple $C_1$ of disjoint solutions. 

    By Observation \ref{obs.3}, and the definition of bottom element, we know that solution $X_2$ found by the algorithm is a predecessor of every feasible solution in $\Gamma(X_1)$; that is, $X_2$ is a predecessor of every feasible solution disjoint from $X_1$. By the same argument as before, $X_2 \cap Y_2 \neq \emptyset$. We then have $X_2 \preceq Y_2 \prec Y_3$, and we may replace $Y_2$ in $C_1$ with $X_2$ to generate a new $\ell$-tuple $C_2$ of disjoint solutions. 

    By repeating this procedure $k$ times, we end up with the $\ell$-tuple $C_k = (X_1, X_2, \ldots, X_k, Y_{k+1}, \ldots, Y_\ell)$ of left-right ordered disjoint solutions. Then, there exists a feasible solution $Y_{k+1}$ that is a strict successor of $X_k$---the last element of tuple $C_\mathrm{ALG}$. But this implies that $X_k$ is disjoint with the top element $X_z$ of $L$, which we know to be false by construction of $C_\mathrm{ALG}$. Thus, we get the necessary contradiction. 
\end{proof}

\paragraph{Time complexity.} The oracles $\oraclemin$ and $\oracleds$ are called $k^*$ times, and $k^*$ is upper bounded by the length of the shortest chain in $P$, which in the worst case, has length $O(n)$. This completes the proof of Theorem \ref{theorem.7}.

\section{Concluding Remarks} \label{sec.7}
We have shown that \textsc{Max-Cov $k$-Diverse Solutions} can be solved in polynomial time by reducing it to submodular function minimization on a distributive lattice, provided that the set $\Gamma$ of feasible solutions satisfies three structural properties. This establishes a general framework for designing polynomial-time algorithms for diverse variants of combinatorial problems.

We applied this framework to \textsc{Minimum $s$-$t$ Cut} and \textsc{Stable Matching}, proving that their respective solution sets meet the required conditions. An interesting direction for future research is to identify additional problems that satisfy these properties and to explore whether they can be exploited in ways that bypass the complexity of submodular function minimization algorithms. We have seen that when restricted to finding only disjoint solutions, there is a general algorithmic framework that bypasses the generality of SFM, but in turn relies on problem specific oracles. 

Our framework focuses primarily on the pairwise sum of Hamming distances as a diversity measure. A natural question is whether other diversity measures can also be used. We showed that the framework extends to both the coverage and absolute-difference diversity measures. Just as we identified structural properties that solution sets must satisfy to enable efficient computation of diverse solutions, we leave it an open problem to characterize a similar set of conditions that diversity measures must fulfill to be compatible with this approach. 

\section*{Acknowledgement}
This research was supported by the European Union’s Horizon 2020 research and innovation programme under the Marie Skłodowska-Curie grant agreement no. 945045, and by the NWO Gravitation project NETWORKS under grant no. 024.002.003.

\bibliography{biblio}

\newcommand{\etalchar}[1]{$^{#1}$}
\begin{thebibliography}{GVPNP21}

\bibitem[BHMM15]{bolandnazar2015note}
Mohammadreza Bolandnazar, Woonghee~Tim Huh, S~Thomas McCORMICK, and Kazuo Murota.
\newblock A note on “order-based cost optimization in assemble-to-order systems”.
\newblock {\em University of Tokyo (February, Techical report}, 2015.

\bibitem[Bir37]{birkhoff1937rings}
Garrett Birkhoff.
\newblock Rings of sets.
\newblock {\em Duke Mathematical Journal}, 3(3):443--454, 1937.

\bibitem[BJM{\etalchar{+}}19]{baste2019fpt}
Julien Baste, Lars Jaffke, Tom{\'a}{\v{s}} Masa{\v{r}}{\'\i}k, Geevarghese Philip, and G{\"u}nter Rote.
\newblock Fpt algorithms for diverse collections of hitting sets.
\newblock {\em Algorithms}, 12(12):254, 2019.

\bibitem[Bla88]{blair1988lattice}
Charles Blair.
\newblock The lattice structure of the set of stable matchings with multiple partners.
\newblock {\em Mathematics of operations research}, 13(4):619--628, 1988.

\bibitem[Bon10]{bonsma2010most}
Paul Bonsma.
\newblock Most balanced minimum cuts.
\newblock {\em Discrete Applied Mathematics}, 158(4):261--276, 2010.

\bibitem[dBMS23]{de2023finding}
Mark de~Berg, Andr{\'e}s~L{\'o}pez Mart{\'\i}nez, and Frits Spieksma.
\newblock Finding diverse minimum st cuts.
\newblock In {\em 34th International Symposium on Algorithms and Computation}, 2023.

\bibitem[DM24]{drabik2024finding}
Karolina Drabik and Tom{\'a}{\v{s}} Masa{\v{r}}{\'\i}k.
\newblock Finding diverse solutions parameterized by cliquewidth.
\newblock {\em arXiv preprint arXiv:2405.20931}, 2024.

\bibitem[DP02]{davey2002introduction}
Brian~A Davey and Hilary~A Priestley.
\newblock {\em Introduction to lattices and order}.
\newblock Cambridge university press, 2002.

\bibitem[Esc72]{escalante1972schnittverbande}
Fernando Escalante.
\newblock Schnittverb{\"a}nde in graphen.
\newblock In {\em Abhandlungen aus dem Mathematischen Seminar der Universit{\"a}t Hamburg}, volume~38, pages 199--220. Springer, 1972.

\bibitem[FGJ{\etalchar{+}}24]{fomin2020diverse}
Fedor~V Fomin, Petr~A Golovach, Lars Jaffke, Geevarghese Philip, and Danil Sagunov.
\newblock Diverse pairs of matchings.
\newblock {\em Algorithmica}, 86(6):2026--2040, 2024.

\bibitem[Gar15]{garg2015introduction}
Vijay~K Garg.
\newblock {\em Introduction to lattice theory with computer science applications}.
\newblock John Wiley \& Sons, 2015.

\bibitem[Gar18]{garg2018applying}
Vijay~K Garg.
\newblock Applying predicate detection to the constrained optimization problems.
\newblock {\em arXiv preprint arXiv:1812.10431}, 2018.

\bibitem[GI89]{gusfield1989stable}
D.~Gusfield and R.W. Irving.
\newblock {\em The Stable Marriage Problem: Structure and Algorithms}.
\newblock Foundations of computing. MIT Press, 1989.

\bibitem[GLS12]{grotschel2012geometric}
Martin Gr{\"o}tschel, L{\'a}szl{\'o} Lov{\'a}sz, and Alexander Schrijver.
\newblock {\em Geometric algorithms and combinatorial optimization}, volume~2.
\newblock Springer Science \& Business Media, 2012.

\bibitem[GM01]{garg2001slicing}
Vijay~K Garg and Neeraj Mittal.
\newblock On slicing a distributed computation.
\newblock In {\em Proceedings 21st International Conference on Distributed Computing Systems}, pages 322--329. IEEE, 2001.

\bibitem[GMRV22]{gangam2022structural}
Rohith~Reddy Gangam, Tung Mai, Nitya Raju, and Vijay~V Vazirani.
\newblock A structural and algorithmic study of stable matching lattices of" nearby" instances, with applications.
\newblock In {\em 42nd IARCS Annual Conference on Foundations of Software Technology and Theoretical Computer Science (FSTTCS 2022)}. Schloss-Dagstuhl-Leibniz Zentrum f{\"u}r Informatik, 2022.

\bibitem[Gra09]{gratzer2009lattice}
George Gratzer.
\newblock {\em Lattice theory: First concepts and distributive lattices}.
\newblock Courier Corporation, 2009.

\bibitem[GVPNP21]{ganesh2021disjoint}
Aadityan Ganesh, HV~Vishwa~Prakash, Prajakta Nimbhorkar, and Geevarghese Philip.
\newblock Disjoint stable matchings in linear time.
\newblock In {\em Graph-Theoretic Concepts in Computer Science: 47th International Workshop, WG 2021, Warsaw, Poland, June 23--25, 2021, Revised Selected Papers 47}, pages 94--105. Springer, 2021.

\bibitem[Hal93]{halin1993lattices}
R~Halin.
\newblock Lattices related to separation in graphs.
\newblock In {\em Finite and Infinite Combinatorics in Sets and Logic}, pages 153--167. Springer, 1993.

\bibitem[Har05]{harzheim2005ordered}
Egbert Harzheim.
\newblock {\em Ordered sets}, volume~7.
\newblock Springer Science \& Business Media, 2005.

\bibitem[HKK{\etalchar{+}}23]{hanaka2022framework}
Tesshu Hanaka, Masashi Kiyomi, Yasuaki Kobayashi, Yusuke Kobayashi, Kazuhiro Kurita, and Yota Otachi.
\newblock A framework to design approximation algorithms for finding diverse solutions in combinatorial problems.
\newblock In {\em Proceedings of the AAAI Conference on Artificial Intelligence}, volume~37, pages 3968--3976, 2023.

\bibitem[HKKO21]{hanaka2021finding}
Tesshu Hanaka, Yasuaki Kobayashi, Kazuhiro Kurita, and Yota Otachi.
\newblock Finding diverse trees, paths, and more.
\newblock In {\em Proceedings of the AAAI Conference on Artificial Intelligence}, volume~35, pages 3778--3786, 2021.

\bibitem[IFF01]{iwata2001combinatorial}
Satoru Iwata, Lisa Fleischer, and Satoru Fujishige.
\newblock A combinatorial strongly polynomial algorithm for minimizing submodular functions.
\newblock {\em Journal of the ACM (JACM)}, 48(4):761--777, 2001.

\bibitem[Jia21]{jiang2021minimizing}
Haotian Jiang.
\newblock Minimizing convex functions with integral minimizers.
\newblock In {\em Proceedings of the 2021 ACM-SIAM Symposium on Discrete Algorithms (SODA)}, pages 976--985. SIAM, 2021.

\bibitem[KC94]{kopka1994d}
Franti{\v{s}}ek K{\^o}pka and Ferdinand Chovanec.
\newblock $ d $-posets.
\newblock {\em Mathematica Slovaca}, 44(1):21--34, 1994.

\bibitem[Knu97]{knuth1997stable}
Donald~Ervin Knuth.
\newblock {\em Stable marriage and its relation to other combinatorial problems: An introduction to the mathematical analysis of algorithms}, volume~10.
\newblock American Mathematical Soc., 1997.

\bibitem[Kum24]{kumabe2024max}
Soh Kumabe.
\newblock Max-distance sparsification for diversification and clustering.
\newblock {\em arXiv preprint arXiv:2411.02845}, 2024.

\bibitem[Mar01]{markowsky2001overview}
George Markowsky.
\newblock An overview of the poset of irreducibles.
\newblock {\em Combinatorial And Computational Mathematics}, pages 162--177, 2001.

\bibitem[Mey82]{meyer1982lattices}
Bernd Meyer.
\newblock On the lattices of cutsets in finite graphs.
\newblock {\em European Journal of Combinatorics}, 3(2):153--157, 1982.

\bibitem[MG01]{mittal2001computation}
Neeraj Mittal and Vijay~K Garg.
\newblock Computation slicing: Techniques and theory.
\newblock In {\em Distributed Computing: 15th International Conference, DISC 2001 Lisbon, Portugal, October 3--5, 2001 Proceedings 15}, pages 78--92. Springer, 2001.

\bibitem[MMR24]{misra2024parameterized}
Neeldhara Misra, Harshil Mittal, and Ashutosh Rai.
\newblock On the parameterized complexity of diverse sat.
\newblock In {\em 35th International Symposium on Algorithms and Computation (ISAAC 2024)}, pages 50--1. Schloss Dagstuhl--Leibniz-Zentrum f{\"u}r Informatik, 2024.

\bibitem[Mur03]{murota2003}
Kazuo Murota.
\newblock {\em Discrete Convex Analysis}.
\newblock Society for Industrial and Applied Mathematics, 2003.

\bibitem[PQ82]{picard1980structure}
Jean-Claude Picard and Maurice Queyranne.
\newblock On the structure of all minimum cuts in a network and applications.
\newblock {\em Math. Program.}, 22(1):121, dec 1982.

\bibitem[RS90]{Roth_Sotomayor_1990}
Alvin~E. Roth and Marilda A.~Oliveira Sotomayor.
\newblock {\em Two-Sided Matching: A Study in Game-Theoretic Modeling and Analysis}.
\newblock Econometric Society Monographs. Cambridge University Press, 1990.

\bibitem[Sch00]{schrijver2000combinatorial}
Alexander Schrijver.
\newblock A combinatorial algorithm minimizing submodular functions in strongly polynomial time.
\newblock {\em Journal of Combinatorial Theory, Series B}, 80(2):346--355, 2000.

\bibitem[SPK{\etalchar{+}}24]{shida2024finding}
Yuto Shida, Giulia Punzi, Yasuaki Kobayashi, Takeaki Uno, and Hiroki Arimura.
\newblock Finding diverse strings and longest common subsequences in a graph.
\newblock In {\em 35th Annual Symposium on Combinatorial Pattern Matching (CPM 2024)}, pages 27--1. Schloss Dagstuhl--Leibniz-Zentrum f{\"u}r Informatik, 2024.

\bibitem[Sta11]{stanley2011enumerative}
Richard~P Stanley.
\newblock Enumerative combinatorics: Volume 1, 2011.

\end{thebibliography}

\end{document}